\documentclass[english]{article}
\usepackage[T1]{fontenc}
\usepackage[latin9]{inputenc}
\usepackage{xcolor}
\usepackage{babel}
\usepackage{array}
\usepackage{calc}
\usepackage{amsmath}
\usepackage{amssymb}
\usepackage{stmaryrd}
\usepackage[authoryear]{natbib}
\PassOptionsToPackage{normalem}{ulem}
\usepackage{ulem}
\usepackage[unicode=true,
 bookmarks=true,bookmarksnumbered=false,bookmarksopen=false,
 breaklinks=false,pdfborder={0 0 1},backref=false,colorlinks=false]
 {hyperref}
\hypersetup{pdftitle={\@shorttitle},
 pdfauthor={\@shortauthor},
 pdfsubject={\@Subject},
 pdfkeywords={Computer Graphics Forum, EUROGRAPHICS},
   colorlinks,linkcolor=blue,citecolor=blue,urlcolor=blue,    bookmarks=false,    pdfpagemode=UseNone}
\usepackage{breakurl}

\makeatletter

\providecommand{\tabularnewline}{\\}
\providecolor{lyxadded}{rgb}{0,0,1}
\providecolor{lyxdeleted}{rgb}{1,0,0}

\DeclareRobustCommand{\lyxsout}[1]{\ifx\\#1\else\sout{#1}\fi}

\usepackage{tikz,tkz-tab,tkz-graph}
\usepackage{amsthm}
\usepackage{tikz}
\usepackage{stackengine}
\setstackgap{S}{1pt}

\newtheorem{prop}{Proposition}[section]
\newtheorem{defin}{Definition}[section]

\newtheorem{claim}{Claim}[section]

\theoremstyle{plain}
\newtheorem{rmk}{Remark}[section]
\newtheorem{example}{Example}[section]
\newtheorem{operation}{Operation}[section]

\usetikzlibrary{arrows,automata,shapes,snakes}
\usetikzlibrary{positioning, fadings}
\tikzset{
    state/.style={
           rectangle,
  fill=#1!5!white,
           rounded corners,
           draw=#1, very thick,
           minimum height=2em,
           inner sep=2pt,
           text centered,
           },
coeff/.style={
           circle,
           draw=black, very thick,
           minimum height=2em,
           inner sep=2pt,
           text centered,
           },	
}

\makeatother

\begin{document}

\title{Adjacency and Tensor Representation in General Hypergraphs.\\
{\normalsize{}Part 2: Multisets, Hb-graphs and Related }$e$-adj{\normalsize{}acency
Tensors}}

\author{Xavier Ouvrard\textsuperscript{1,2}\quad{}Jean-Marie Le Goff\textsuperscript{1}\quad{}Stéphane
Marchand-Maillet\textsuperscript{2}\\
\\
1~:~CERN\qquad{}2~:~University of Geneva\\
{\small{}\{xavier.ouvrard\}@cern.ch}}
\maketitle
\begin{abstract}
HyperBagGraphs (hb-graphs as short) extend hypergraphs by allowing
the hyperedges to be multisets. Multisets are composed of elements
that have a multiplicity. When this multiplicity has positive integer
values, it corresponds to non ordered lists of potentially duplicated
elements. We define hb-graphs as family of multisets over a vertex
set; natural hb-graphs correspond to hb-graphs that have multiplicity
functions with positive integer values. Extending the definition of
$e$-adjacency to natural hb-graphs, we define different way of building
an $e$-adjacency tensor, that we compare before having a final choice
of the tensor. This hb-graph $e$-adjacency tensor is used with hypergraphs.
\end{abstract}

\section{Introduction}

Hypergraphs were introduced in \citet{berge1973graphs}. Hypergraphs
are defined as a family of nonempty subsets - called hyperedges -
of the set of vertices. Elements of a set are unique. Hence elements
of a given hyperedge are also unique in a hypergraph.

Multisets extend sets by allowing duplication of elements. As mentioned
in \citet{singh2007overview}, N.G. de Bruijn proposed to Knuth the
terminology multiset in replacement of a variety of existing terms,
such as bag or weighted set. Multisets are used in database modelling:
in \citet{albert1991algebraic} relational algebra extension basements
were introduced to manipulate bags - see also \citet{klug1982equivalence}
- by studying bags algebraic properties. Queries for such bags have
been largely studied in a series of articles - see references in \citet{grumbach1996query}:
bags are intensively used in database queries as duplicate search
is a costly operation. In \citet{hernich2017foundations} information
integration under bag semantics is studied as well as the tractability
of some algorithmic problems over bag semantics: they showed that
the GLAV (Global-And-Local-As-View) mapping of two databases problem
becomes untractable over such semantic. In \citet{radoaca2015properties}
and \citet{radoaca2015simple}, the author extensively study multisets
and propose to represent them by two kind of Venn diagrams. Multisets
are also used in P-computing in the form of labelled multiset, called
membrane - see \citet{puaun2006introduction} for more details.

Taking advantage of this duplication allowance, we construct in this
article an extension of hypergraphs called hyper-bag-graphs (shortcut
as hb-graphs). There are two main reasons to get such an extension.
The first reason is that multisets are extensively used in databases
as they allow presence of duplicates \citet{lamperti2000multisets}
- removing duplicates (and thus obtaining sets and hypergraphs) being
an expensive operation. The second reason is that natural hb-graphs
- hb-graphs based on multisets with non-negative integer multiplicity
values - allow results on the hb-graph adjacency tensor: hypergraph
being particular case of hb-graph, other hypergraph $e$-adjacency
tensors than the ones proposed in \citet{banerjee2017spectra,ouvrard2017cooccurrence}
can be built by giving meaningful interpretation to the steps taken
during its construction via hb-graph. 

Section \ref{sec:Background-and-related} gives the mathematical background,
including main definitions on hypergraphs and multisets. Section \ref{sec:hb-graphs}
gives mathematical construction of the Hyper-Bag-Graphs (or hb-graphs).
Section \ref{sec:algebraic desc} gives algebraic description of hb-graphs
and consequences for the adjacency tensor of hypergraphs. Section
\ref{sec:Results_constructed_tensors} gives results on the constructed
tensors. Section \ref{sec:Evaluation} evaluates the constructed tensors
and proceed to a final choice on the hypergraph $e$-adjacency tensor.
Section \ref{sec:Future-work-and} gives future work.

\section{Mathematical background}

\label{sec:Background-and-related}

\subsection{Hypergraphs}

As mentioned in \citet{ouvrard2017cooccurrence}, hypergraphs fit
collaboration networks modelling - \citet{newman2001scientific,newman2001scientific-2}
-, co-author networks - \citet{grossman1995portion}, \citet{taramasco2010academic}
-, chemical reactions - \citet{temkin1996chemical} - , genome - \citet{chauve2013hypergraph}
-, VLSI design - \citet{karypis1999multilevel} - and other applications.
More generally hypergraphs fit perfectly to keep entities grouping
information. Hypergraphs succeed in capturing $p$-adic relationships.
In \citet{berge1973graphs}, \citet{stell2012relations} and \citet{bretto2013hypergraph}
hypergraphs are defined in different ways. In this article, the definition
of \citet{bretto2013hypergraph} - as it doesn't impose the union
of the hyperedges to cover the vertex set - is used:

\begin{defin}

An \textbf{(undirected)} \textbf{hypergraph} $\mathcal{H}=\left(V,E\right)$
on a finite set of $n$ vertices (or vertices) $V=\left\{ v_{j}:j\in\left\llbracket n\right\rrbracket \right\} $
is defined as a family of $p$ \textbf{hyperedges} $E=\left\{ e_{j}:j\in\left\llbracket p\right\rrbracket \right\} $
where each hyperedge is a non-empty subset of $V$.

A \textbf{weighted hypergraph} is a triple: $\mathcal{H}_{w}=\left(V,E,w\right)$
where $\mathcal{H}=\left(V,E\right)$ is a hypergraph and $w$ a mapping
where each hyperedge $e\in E$ is associated to a real number $w(e)$.

The \textbf{$2$-section} of a hypergraph $\mathcal{H}=\left(V,E\right)$
is the graph $\left[\mathcal{H}\right]_{2}=\left(V,E'\right)$ such
that:
\[
\forall u\in V,\forall v\in V\,:\,(u,v)\in E'\Leftrightarrow\exists e\in E\,:\,u\in e\land v\in e
\]

Let $k\in\mathbb{N}^{*}$. A hypergraph is \textbf{$k$-uniform} if
all its hyperedges have the same cardinality $k$. 

A \textbf{directed hypergraph} $\mathcal{H}=\left(V,E\right)$ is
a hypergraph where each hyperedge $e_{i}\in E$ accepts a partition
in two non-empty subsets, called the \textbf{source} - written $e_{s\,i}$
- and the \textbf{target} - written $e_{t\,i}$ - with $e_{s\,i}\cap e_{t\,i}=\emptyset$

\end{defin}

\begin{defin}

Let $\mathcal{H}=\left(V,E\right)$ be a hypergraph.

The \textbf{degree} of a vertex is the number of hyperedges it belongs
to. For a vertex $v_{i}$, it is written $d_{i}$ or $\deg\left(v_{i}\right)$.
It holds: $d_{i}=\left|\left\{ e\,:\,v_{i}\in e\right\} \right|$

\end{defin}

In this article only undirected hypergraphs will be considered. Hyperedges
link one or more vertices together. Broadly speaking, the role of
the hyperedges in hypergraphs is playing the role of edges in graphs.

\subsection{Multisets}

\subsubsection{Generalities}

Basic on multisets are given in this section, based mainly on \citet{singh2007overview}.

\begin{defin}

Let $A$ be a set of distinct objects. Let $\mathbb{W}\subseteq\mathbb{R}^{+}$

Let $m$ be an application from $A$ to $\mathbb{W}$.

Then $A_{m}=\left(A,m\right)$ is called a \textbf{multiset} - or
\textbf{mset} or \textbf{bag} - on $A$. 

$A$ is called the \textbf{ground} or the \textbf{universe }of the
multiset $A_{m}$, $m$ is called the \textbf{multiplicity function}
of the multiset $A_{m}$.

$A_{m}^{\star}=\left\{ x\in A\colon m(x)\neq0\right\} $ is called
the \textbf{support} - or \textbf{root} or \textbf{carrier} - of $A_{m}$.

The elements of the support of a mset are called its \textbf{generators}.

A multiset where $\mathbb{W}\subseteq\mathbb{N}$ is called a \textbf{natural
multiset}.

\end{defin}

We write $\mathcal{M}\left(A\right)$ the set of all multisets of
universe $A$.

Some extensions of multisets exist where the multiplicity function
can have its range in $\mathbb{Z}$ - called hybrid set in \citet{loeb1992sets}.
Some other extensions exist like fuzzy multisets \citet{syropoulos2000mathematics}. 

Several notations of msets exist. One common notation which we will
use, if $A=\left\{ x_{j}:j\in\left\llbracket n\right\rrbracket \right\} $
is the ground of a mset $A_{m}$ is to write: 
\[
A_{m}=\left\{ x_{i}^{m_{i}}:i\in\left\llbracket n\right\rrbracket \right\} 
\]
 where $m_{i}=m\left(x_{i}\right)$. 

An other notation is: 
\[
\left\{ x_{1},\ldots,x_{n}\right\} _{m_{1},...,m_{n}}
\]
 or even: 
\[
m_{1}\left\{ x_{1}\right\} +\ldots+m_{n}\left\{ x_{n}\right\} .
\]

If $A_{m}$ is a natural multiset an other notation is: 
\[
\left\{ \left\{ \underset{m_{1}\,\text{times}}{\underbrace{x_{1},\ldots,x_{1}}},\ldots,\underset{m_{n}\,\text{times}}{\underbrace{x_{n},\ldots,x_{n}}}\right\} \right\} 
\]
which is similar to have an unordered list.

\begin{rmk}
\begin{enumerate}
\item Two msets can have same support and same support objects multiplicities
but can differ by their universe.\\
Also to be equal two msets must have same universe, same support and
same multiplicity function.
\item The multiplicity function corresponds to a weight that is associated
to objects of the universe.
\item Multiplicity in natural multisets can also be interpreted as a duplication
of support elements. In this case, a mset can be viewed as a non ordered
list with repetition. In a natural multiset the copies of a generator
$a$ of the support in $m(a)$ instances are called \textbf{elements}
of the multiset.
\item Some definitions of multisets also consider $\mathbb{W}=\mathbb{R}$
which could lead to interesting applications. We don't develop such
case here.
\end{enumerate}
\end{rmk}

\begin{defin}

Let $A_{m}$ be a mset. 

The \textbf{m-cardinality} of $A_{m}$ written $\#_{m}A_{m}$ is defined
as:

\[
\#_{m}A_{m}=\sum\limits _{x\in A}m(x).
\]

The \textbf{cardinality} of $A_{m}$ - written $\#A_{m}$ is defined
as:

\[
\#A_{m}=\left|A_{m}^{\star}\right|.
\]

\end{defin}

\begin{rmk}

In general multisets, m-cardinality and cardinality are two separated
notions as for instance: $A=\left\{ a^{1.2},b^{0.8}\right\} $, $B=\left\{ a^{0.2},b^{1.8}\right\} $
and, $C=\left\{ a^{0.5},b^{0.5}\right\} $ have all same cardinality
with different m-cardinalities for C compared to A and B.

In natural multisets, m-cardinality and cardinality are equal if and
only if the multiplicity of each element in the support is 1, ie if
the natural multiset is a set. It doesn't generalize to general multisets
- see A and B of the former example.

\end{rmk}

\begin{defin}

Two msets $A_{m_{1}}$ and $B_{m_{2}}$ are said to be \textbf{cognate}
if they have same support.

\end{defin}

They are not necessarily equal: for instance, $\left\{ a^{1},b^{2}\right\} $
and $\left\{ a^{2},b^{1}\right\} $ are cognate but different.

\begin{defin}

Let $\mathcal{A}=U_{m_{\mathcal{A}}}$ and $\mathcal{B}=U_{m_{\mathcal{B}}}$
be two msets on the same universe $U$.

If $\mathcal{A}^{\star}=\emptyset$ $\mathcal{A}$ is called the \textbf{empty
mset} and written $\emptyset$.

$\mathcal{A}$ is said to be \textbf{included} in $\mathcal{B}$ -
written $\mathcal{A}\subseteq\mathcal{B}$ - if for all $x\in U$:
$m_{\mathcal{A}}(x)\leqslant m_{\mathcal{B}}(x)$. In this case, $\mathcal{A}$
is called a \textbf{submset} of $\mathcal{B}$.

The \textbf{union} of $\mathcal{A}$ and $\mathcal{B}$ is the mset
$\mathcal{C}=\mathcal{A}\cup\mathcal{B}$ of universe $U$ and of
multiplicity function $m_{\mathcal{C}}$ such that for all $x\in U$:
\[
m_{\mathcal{C}}(x)=\max\left(m_{\mathcal{A}}(x),m_{\mathcal{B}}(x)\right).
\]

The \textbf{intersection} of $\mathcal{A}$ and $\mathcal{B}$ is
the mset $\mathcal{D}=\mathcal{A}\cap\mathcal{B}$ of universe $U$
and of multiplicity function $m_{\mathcal{D}}$ such that for all
$x\in U$: 
\[
m_{\mathcal{D}}(x)=\min\left(m_{\mathcal{A}}(x),m_{\mathcal{B}}(x)\right).
\]

The \textbf{sum} of $\mathcal{A}$ and $\mathcal{B}$ is the mset
$\mathcal{E}=\mathcal{A}\uplus\mathcal{B}$ of universe $U$ and of
multiplicity function $m_{\mathcal{E}}$ such that for all $x\in U$:
\[
m_{\mathcal{E}}(x)=m_{\mathcal{A}}(x)+m_{\mathcal{B}}(x).
\]

\end{defin}

\begin{prop}$\cup$, $\cap$ and $\uplus$ are commutative and associative
laws on msets of same universe. They have the empty mset of same universe
as identity law.

$\uplus$ is distributive for $\cup$ and $\cap$.

$\cup$ and $\cap$ are distributive one for the other.

$\cup$ and $\cap$ are idempotent.

\end{prop}

\begin{defin}

Let $A$ be a mset.

The \textbf{power set} of $A$, written $\widetilde{\mathcal{P}}(A)$,
is the multiset of all submsets of $A$.

\end{defin}

\subsubsection{Copy-set of a multiset}

Let consider a multiset: $A_{m}=\left(A,m\right)$ where the range
of the multiplicity function is a subset of $\mathbb{N}$. Equivalent
definition - see \citet{syropoulos2000mathematics} - is to give a
couple $<A_{0},\rho>$ where $A_{0}$ is the set of all instances
(including copies) of $A_{m}$ with an equivalency relation $\rho$
where: 
\[
\forall x\in A_{0},\forall x'\in A_{0}:\,\,\,x\rho x'\Leftrightarrow\exists!c\in A:\,x\rho c\land x'\rho c.
\]

\begin{defin}

Two elements of $A_{0}$ such that: $x\rho x'$ are said copies one
of the other. The unique $c\in A$ is called the original element.
$x$ and $x'$ are said copies of $c$.

\end{defin}

Also $A_{0}/\rho$ is isomorphic to $A$ and: 
\[
\forall\overline{x}\in A_{0}/\rho,\exists!c\in A:\left|\left\{ x:x\in\overline{x}\right\} \right|=m(c)\land\forall x\in\overline{x}:x\rho c.
\]

\begin{defin}

The set $A_{0}$ is called a \textbf{copy-set} of the multiset $A_{m}$.

\end{defin}

\begin{rmk}

A copy-set for a given multiset is not unique. Sets of equivalency
classes of two couples $<A_{0},\rho>$ and $<A_{0}^{\prime},\rho^{\prime}>$
of a given multiset are isomorphic.

\end{rmk}

\subsubsection{Algebraic representation of a multiset}

We suppose given a natural multiset $A_{m}=\left(A,m\right)$ of universe
$A=\left\{ \alpha_{i}:i\in\left\llbracket n\right\rrbracket \right\} $
and multiplicity function $m$. It yields: 
\[
A_{m}=\left\{ \alpha_{i_{j}}^{m\left(\alpha_{i_{j}}\right)}:\alpha_{i_{j}}\in A_{m}^{\star}\right\} .
\]

\paragraph*{Vector representation: }

A multiset can be conveniently represented by a vector of length the
cardinality of the universe and where the coefficients of the vector
represent the multiplicity of the corresponding element.

\begin{defin}

The \textbf{vector representation of the multiset} $A_{m}$ is the
vector $\overrightarrow{A}=\left(m\left(\alpha\right)\right)_{\alpha\in A}.$

\end{defin}

This representation requires $\left|A\right|$ space and has $\left|A\right|-\left|A_{m}^{\star}\right|$
null elements.

The sum of the elements of $\overrightarrow{A}$ is $\sharp_{m}A_{m}$

This representation will be useful later when considering family of
multisets in order to build the incident matrix.

\paragraph*{Hypermatrix representation: }

An alternative representation is built by using a symmetric hypermatrix.
This approach is needed to reach our goal of constructing an $e$-adjacency
tensor for general hypergraphs.

\begin{defin}

The \textbf{unnormalized hypermatrix representation of the multiset
$A_{m}$} is the symmetrix hypermatrix $A_{u}=\left(a_{u,i_{1}...i_{r}}\right)_{\left(i_{1},...,i_{r}\right)\in\left\llbracket n\right\rrbracket }$
of order $r=\sharp_{m}A_{m}$ and dimension $n$ such that $a_{u,i_{1}...i_{r}}=1$
if $\forall j\in\left\llbracket r\right\rrbracket :i_{j}\in\left\llbracket n\right\rrbracket \land\alpha_{i_{j}}\in A_{m}^{\star}$.
The other elements are null.

\end{defin}

Hence the number of non-zero elements in $A_{u}$ is $\dfrac{r!}{\prod\limits _{\alpha\in A_{m}^{\star}}m\left(\alpha\right)}$
out of the $n^{r}$ elements of the representation. 

The sum of the elements of $A_{u}$ is then: $\dfrac{r!}{\prod\limits _{\alpha\in A_{m}^{\star}}m\left(\alpha\right)}.$

To achieve a normalisation, we enforce the sum of the elements of
the hypermatrix to be the m-rank of the multiset it encodes. It yields:

\begin{defin}

The \textbf{normalized hypermatrix representation of the multiset
$A_{m}$} is the symmetrix hypermatrix $A=\left(a_{i_{1}...i_{r}}\right)_{\left(i_{1},...,i_{r}\right)\in\left\llbracket n\right\rrbracket }$
of order $r=\sharp_{m}A_{m}$ and dimension $n$ such that $a_{i_{1}...i_{r}}=\dfrac{\prod\limits _{\alpha\in A_{m}^{\star}}m\left(\alpha\right)}{\left(r-1\right)!}$
if $\forall j\in\left\llbracket r\right\rrbracket :i_{j}\in\left\llbracket n\right\rrbracket \land\alpha_{i_{j}}\in A_{m}^{\star}$.
The other elements are null.

\end{defin}

\section{Hb-graphs}

\label{sec:hb-graphs}

Hyper-bag-graphs - hb-graphs for short - are introduced in this section.
Hb-graphs extend hypergraphs by allowing hyperedges to be msets. The
goal of this section is to revisit some of the definitions and results
found in \citet{bretto2013hypergraph} for hypergraphs and extend
them to hb-graphs.

\subsection{Generalities}

\subsubsection{First definitions}

\begin{defin}

Let $V=\left\{ v_{i}:i\in\left\llbracket n\right\rrbracket \right\} $
be a nonempty finite set.

A \textbf{hyper-bag-graph} - or \textbf{hb-graph} - is a family of
msets with universe $V$ and support a subset of $V$. The msets are
called the \textbf{hb-edges} and the elements of $V$ the \textbf{vertices}.

We write $E=\left(e_{i}\right)_{i\in\left\llbracket p\right\rrbracket }$
the family of hb-edges and $\mathcal{H}=\left(V,E\right)$ such a
hb-graph.

\end{defin}

We consider for the remainder of the article a hb-graph $\mathcal{H}=\left(V,E\right)$,
with $V=\left\{ v_{i}:i\in\left\llbracket n\right\rrbracket \right\} $
and $E=\left(e_{i}\right)_{i\in\left\llbracket p\right\rrbracket }$
the family of its hb-edges.

Each hb-edge $e_{i}\in E$ is of universe $V$ and has a multiplicity
function associated to it: $m_{e_{i}}:V\rightarrow\mathbb{W}$ where
$\mathbb{W}\subseteq\mathbb{R}^{+}$. When the context make it clear
the notation $m_{i}$ is used for $m_{e_{i}}$ and $m_{ij}$ for $m_{e_{i}}\left(v_{j}\right)$.

\begin{defin}

A hb-graph is said with \textbf{no repeated hb-edges} if: 
\[
\forall i_{1}\in\left\llbracket p\right\rrbracket ,\forall i_{2}\in\left\llbracket p\right\rrbracket :e_{i_{1}}=e_{i_{2}}\Rightarrow i_{1}=i_{2}.
\]

\end{defin}

\begin{defin}

A hb-graph where each hb-edge is a natural mset is called a \textbf{natural
hb-graph}.

\end{defin}

\begin{rmk}

For a general hb-graph each hb-edge has to be seen as a weighted system
of vertices, where the weights of each vertex are hb-edge dependent.

In a natural hb-graph the multiplicity function can be viewed as a
duplication of the vertices.

\end{rmk}

\begin{defin}

The \textbf{order} of a hb-graph $\mathcal{H}$ - written $O\left(\mathcal{H}\right)$
- is: 
\[
O\left(\mathcal{H}\right)=\sum\limits _{j\in\left\llbracket n\right\rrbracket }\underset{e\in E}{\max}\left(m_{e}\left(v_{j}\right)\right).
\]

Its \textbf{size} is the cardinality of $E.$

\end{defin}

\begin{defin}

The \textbf{empty hb-graph} is the hb-graph with an empty set of vertices.

The \textbf{trivial hb-graph} is the hb-graph with a non empty set
of vertices and an empty family of hb-edges.

\end{defin}

If : $\bigcup\limits _{i\in\left\llbracket p\right\rrbracket }e_{i}^{\star}=V$
then the hb-graph is said with no isolated vertices. Otherwise, the
elements of $V\backslash\bigcup\limits _{i\in\left\llbracket p\right\rrbracket }e_{i}^{\star}$
are called the \textbf{isolated vertices}. They correspond to elements
of hyperedges which have zero-multiplicity for all hb-edges.

\begin{rmk}

A \textbf{hypergraph} is a natural hb-graph where the vertices of
the hb-edges have multiplicity one for any vertex of their support
and zero otherwise.

\end{rmk}

\subsubsection{Support hypergraph}

\begin{defin}

The \textbf{support hypergraph} of a hb-graph $\mathcal{H}=\left(V,E\right)$
is the hypergraph whose vertices are the ones of the hb-graph and
whose hyperedges are the support of the hb-edges in a one-to-one way.
We write it $\underline{\mathcal{H}}=\left(V,\underline{E}\right)$,
where $\underline{E}=\left\{ e^{\star}:e\in E\right\} $.

\end{defin}

\begin{rmk}

Given a hypergraph, an infinite set of hb-graphs can be generated
that all have this hypergraph as support. To each of these hb-graphs
corresponds a hb-edge family: to each support of these hb-edges corresponds
at least a hyperedge in the hypergraph and reciprocally to each hyperedge
corresponds at least a hb-edge in each hb-graph of the infinite set.

To have unicity, the considered hypergraph and hb-graphs should be
respectively with no repeated hyperedge or with no repeated hb-edge.

\end{rmk}

\subsubsection{m-uniform hb-graphs}

\begin{defin}

The \textbf{m-range} of a hb-graph - written $r_{m}\left(\mathcal{H}\right)$
- is by definition:

\[
r_{m}\left(\mathcal{H}\right)=\underset{e\in E}{\max}\#_{m}e.
\]

The \textbf{range} of a hb-graph $\mathcal{H}$ - written $r\left(\mathcal{H}\right)$
- is the range of its support hypergraph $\underline{\mathcal{H}}.$

The \textbf{m-co-range} of a hb-graph - written $cr_{m}\left(\mathcal{H}\right)$
- is by definition:

\[
cr_{m}\left(\mathcal{H}\right)=\underset{e\in E}{\min}\#_{m}e.
\]

The \textbf{co-range} of a hb-graph $\mathcal{H}$ - written $cr\left(\mathcal{H}\right)$
- is the range of its support hypergraph $\underline{\mathcal{H}}.$

\end{defin}

\begin{defin}

A hb-graph is said $k$-m-uniform if all its hb-edges have same $m$-cardinality
$k$.

A hb-graph is said $k$-uniform if its support hypergraph is $k$-uniform.

\end{defin}

\begin{prop}

A hb-graph $\mathcal{H}$ is $k$-m-uniform if and only if: 
\[
r_{m}\left(\mathcal{H}\right)=cr_{m}\left(\mathcal{H}\right)=k.
\]

\end{prop}

\begin{proof}

Immediate.

\end{proof}

\subsubsection{HB-star and m-degree}

\begin{defin}

The \textbf{HB-star} of a vertex $x\in V$ is the multiset - written
$H(x)$ - defined as:

\[
H(x)=\left\{ e^{m_{e}(x)}\,:\,e\in E\land x\in e^{*}\right\} .
\]

\end{defin}

\begin{rmk}

The support of the HB-star $H^{*}(x)$ of a vertex $x\in V$ of a
hb-graph $\mathcal{H}$ is exactly the star of this vertex in the
support hypergraph $\underline{\mathcal{H}}$.

\end{rmk}

\begin{defin}

The \textbf{m-degree of a vertex} $x\in V$ of a hb-graph $\mathcal{H}$
- written $\deg_{m}\left(x\right)=d_{m}(x)$ - is defined as:

\[
\deg_{m}(x)=\#_{m}H(x).
\]

The \textbf{maximal m-degree} of a hb-graph $\mathcal{H}$ is written
$\Delta_{m}=\underset{x\in V}{\max}\,\deg_{m}(x)$.

The degree of a vertex $x\in V$ of a hb-graph $\mathcal{H}$ - written
$\deg\left(x\right)=d(x)$ - corresponds to the degree of this vertex
in the support hypergraph $\underline{\mathcal{H}}.$

The maximal degree of a hb-graph $\mathcal{H}$ is written $\Delta$
and corresponds to the maximal degree of the support hypergraph $\underline{\mathcal{H}}.$

\end{defin}

\begin{defin}

A hb-graph having all of its hb-edges of same m-degree $k$ is said
\textbf{m-regular} or $k$-m-regular.

A hb-graph is said \textbf{regular} if its support hypergraph is regular.

\end{defin}

\subsubsection{Dual of a hb-graph}

\begin{defin}

Considering a hb-graph $\mathcal{H}$, its dual is the hb-graph $\tilde{\mathcal{H}}$
with a set of vertices $\tilde{V}=\left\{ \tilde{x_{i}}:i\in\left\llbracket p\right\rrbracket \right\} $
which is in bijection $f$ with the set of hb-edges $E$ of $\mathcal{H}$:

\[
\forall\tilde{x_{i}}\in\tilde{V},\exists!e_{i}\in E:\,\tilde{x_{i}}=f\left(e_{i}\right).
\]

And the set of hb-edges $\tilde{E}=\left\{ \tilde{e_{j}}:j\in\left\llbracket n\right\rrbracket \right\} $
is in bijection $g:x_{j}\mapsto\tilde{e_{j}}$ - where $\tilde{e_{j}}=\left\{ \tilde{x_{i}}^{m_{e_{i}}\left(x_{j}\right)}\colon i\in\left\llbracket p\right\rrbracket \land\tilde{x_{i}}=f\left(e_{i}\right)\land x_{j}\in e_{i}^{*}\right\} $
- with the set of vertices of $\mathcal{H}$.

\end{defin}

Switching from the hb-graph to its dual:

\begin{center}%
\begin{tabular}{|c|c|c|}
\hline 
 & $\mathcal{H}$ & $\widetilde{\mathcal{H}}$\tabularnewline
\hline 
Vertices & $x_{i},i\in\left\llbracket n\right\rrbracket $ & $\tilde{x_{j}}=f\left(e_{j}\right),j\in\left\llbracket p\right\rrbracket $\tabularnewline
\hline 
Edges & $e_{j},j\in\left\llbracket p\right\rrbracket $ & $\tilde{e_{i}}=g\left(x_{i}\right),i\in\left\llbracket n\right\rrbracket $\tabularnewline
\hline 
Multiplicity & $x_{i}\in e_{j}$ with $m_{e_{j}}\left(x_{i}\right)$ & $\tilde{x_{j}}\in\tilde{e_{i}}$ with $m_{e_{i}}\left(x_{j}\right)$\tabularnewline
\hline 
 & $d_{m}\left(x_{i}\right)$ & $\#_{m}\tilde{e_{i}}$\tabularnewline
\hline 
 & $\#_{m}e_{i}$ & $d_{m}\left(\tilde{x_{j}}\right)$\tabularnewline
\hline 
 & $k$-m-uniform & $k$-m-regular\tabularnewline
\hline 
 & $k$-m-regular & $k$-m-uniform\tabularnewline
\hline 
\end{tabular}\end{center}

\subsection{Additional concepts for natural hb-graphs}

\subsubsection{Numbered copy hypergraph of a natural hb-graph}

In natural hb-graphs the hb-edge multiplicity functions have their
range in the natural number set. The vertices in a hb-edge with multiplicities
strictly greater than 1 can be seen as copies of the original vertex.

Deepening this approach copies have to be understood as ``numbered''
copies. Let $A$ and $B$ be two hb-edges. Let $v_{i}$ be a vertex
of multiplicity $m_{A}$ in $A$ and $m_{B}$ in $B$. $A\cap B$
will hold $\min\left(m_{A},m_{B}\right)$ copies: the ones ``numbered''
from 1 to $\min\left(m_{A},m_{B}\right)$. The remaining copies will
be held either in $A$ xor $B$ depending which set has the highest
multiplicity of $v_{i}$.

More generally, we define the numbered-copy set of a multiset:

\begin{defin}

Let $A_{m}=\left\{ x_{i}^{m_{i}}:i\in\left\llbracket n\right\rrbracket \right\} $.

The \textbf{numbered copy-set} of $A_{m}$ is the copy-set $\breve{A_{m}}=\left\{ \left[x_{i\,j}\right]_{m_{i}}:i\in\left\llbracket n\right\rrbracket \right\} $
where: $\left[x_{i\,j}\right]_{m_{i}}$ is a shortcut to indicate
the numbered copies of the original element $x_{i}$: $x_{i\,1}$
to $x_{i\,m_{i}}$ and $j$ is designated as the copy number of the
element $x_{i}$.

\end{defin}

\begin{defin}

Let $\mathcal{H}=\left(V,E\right)$ be a natural hb-graph. 

Let $V=\left\{ v_{j}:j\in\left\llbracket n\right\rrbracket \right\} $
be the vertices of the hb-graph. Let $E=\left\{ e_{k}\colon\,k\in\left\llbracket p\right\rrbracket \right\} $
be the hb-edges of the hb-graph and for $k\in\left\llbracket p\right\rrbracket $,
$m_{e_{k}}$ the multiplicity function of $e_{k}\in E$.

The \textbf{maximum multiplicity function} of $\mathcal{H}$ is the
function $m:V\rightarrow\mathbb{N}$ defined for all $v\in V$ by:
\[
m(v)=\underset{e\in E}{\max}\,m_{e}(v).
\]

\end{defin}

\begin{defin}

Let $\mathcal{H}=\left(V,E\right)$ be a natural hb-graph where $V=\left\{ v_{i}:i\in\left\llbracket n\right\rrbracket \right\} $
is the vertex set and $E=\left(e_{k}\right)_{k\in\left\llbracket p\right\rrbracket }$
is the hb-edge family of the hb-graph.

Let $m$ be the maximum multiplicity function.

Let consider the numbered-copy-set of the multiset $\left\{ v_{i}^{m\left(v_{i}\right)}:i\in\left\llbracket n\right\rrbracket \right\} $:
\[
\breve{V}=\left\{ \left[v_{i\,j}\right]_{m\left(v_{i}\right)}:i\in\left\llbracket n\right\rrbracket \right\} .
\]

Then each hb-edge $e_{k}=\left\{ v_{i_{j}}^{m_{k\,i_{j}}}:j\in\left\llbracket k\right\rrbracket \land i_{j}\in\left\llbracket n\right\rrbracket \right\} $
is associated to a copy-set / equivalency relation $<e_{k\,0},\rho_{k}>$
which elements are in $\breve{V}$ with copy number as small as possible
for each vertex in $e_{k}$.

Then $\mathcal{H}_{0}=\left(\breve{V},E_{0}\right)$ where $E_{0}=\left\{ e_{k\,0}\colon\,k\in\left\llbracket p\right\rrbracket \right\} $
is a hypergraph called the \textbf{numbered-copy-hypergraph} of $\mathcal{H}$.

\end{defin}

\begin{prop}

A numbered-copy-hypergraph is unique for a given hb-graph.

\end{prop}

\begin{proof}

It is immediate by the way the numbered-copy-hypergraph is built from
the hb-graph.

\end{proof}

Allowing the duplicates to be numbered prevent ambiguities; nonetheless
it has to be seen as a conceptual approach as duplicates are entities
that are not discernible.

\subsubsection{Paths, distance and connected components}

Defining a path in a hb-graph is not straightforward as vertices are
duplicated in a hb-graph. The duplicate of a vertex strictly inside
a path must be at the intersection of two consecutive hb-edges.

\begin{defin}

A \textbf{strict m-path} $x_{0}e_{1}x_{1}\ldots e_{s}x_{s}$ in a
hb-graph from a vertex $x$ to a vertex $y$ is a vertex / hb-edge
alternation with hb-edges $e_{1}$ to $e_{s}$ and vertices $x_{0}$
to $x_{s}$ such that $x_{0}=x$, $x_{s}=y$, $x\in e_{1}$ and $y\in e_{s}$
and that for all $i\in\left\llbracket s-1\right\rrbracket $, $x_{i}\in e_{i}\cap e_{i+1}$.

A \textbf{large m-path} $x_{0}e_{1}x_{1}\ldots e_{s}x_{s}$ from a
vertex $x$ to a vertex $y$ is a vertex / hb-edge alternation with
hb-edges $e_{1}$ to $e_{s}$ and vertices $x_{0}$ to $x_{s}$ such
that $x_{0}=x$, $x_{s}=y$, $x\in e_{1}$ and $y\in e_{s}$ and that
for all $i\in\left\llbracket s-1\right\rrbracket $, $x_{i}\in e_{i}\cup e_{i+1}$.

$s$ is called in both cases the length $l(x,y)$ of the m-path from
$x$ to $y$.

Vertices from $x_{1}$ to $x_{s-1}$ are called \textbf{interior vertices}
of the m-path.

$x_{0}$ and $x_{s}$ are called \textbf{extremities} of the m-path.

If the extremities are different copies of the same object, then the
m-path is said to be an \textbf{almost cycle}.

If the extremities designate exactly the same copy of one object,
the m-path is said to be a \textbf{cycle}.

\end{defin}

\begin{rmk}
\begin{enumerate}
\item For a strict m-path, there are:
\[
\prod\limits _{i\in\left\llbracket s-1\right\rrbracket }m_{e_{i}\cap e_{i+1}}\left(x_{i}\right)
\]
possibilities of choosing the interior vertices along a given m-path
$x_{0}e_{1}x_{1}\ldots e_{s}x_{s}$ and: 
\[
m_{e_{1}}\left(x_{0}\right)\prod\limits _{i\in\left\llbracket s-1\right\rrbracket }m_{e_{i}\cap e_{i+1}}\left(x_{i}\right)m_{e_{s}}\left(x_{s}\right)
\]
possible strict m-paths in between the extremities.
\item For a large m-path, there are: 
\[
\prod\limits _{i\in\left\llbracket s-1\right\rrbracket }m_{e_{i}\cup e_{i+1}}\left(x_{i}\right)
\]
possibilities of choosing the interior vertices along a given m-path
$x_{0}e_{1}x_{1}\ldots e_{s}x_{s}$ and:
\[
m_{e_{1}}\left(x_{0}\right)\prod\limits _{i\in\left\llbracket s-1\right\rrbracket }m_{e_{i}\cup e_{i+1}}\left(x_{i}\right)m_{e_{s}}\left(x_{s}\right)
\]
possible large m-paths in between the extremities.
\item As large m-paths between two extremities by a given sequence of interior
vertices and hb-edges include strict m-paths, we often refer as \textbf{m-paths}
for large m-paths.
\item If an m-path exists from $x$ to $y$ then an m-path also exists from
$y$ to $x$.
\end{enumerate}
\end{rmk}

\begin{defin}

An m-path $x_{0}e_{1}x_{1}\ldots e_{s}x_{s}$ in a hb-graph corresponds
to a unique path in the hb-graph support hypergraph called the \textbf{support
path}.

\end{defin}

\begin{prop}

Every m-path $x_{0}e_{1}x_{1}\ldots e_{s}x_{s}$ traversing same hyperedges
and having similar copy vertices as intermediate and extremity vertices
share the same support path.

\end{prop}

The notion of distance is similar to the one defined for hypergraphs.

\begin{defin}

Let $x$ and $y$ be two vertices of a hb-graph. The distance $d(x,y)$
from $x$ to $y$ is the minimal length of an m-path from $x$ to
$y$ if such an m-path exists. If no m-path exist, $x$ and $y$ are
said disconnected and $d(x,y)=+\infty$. 

\end{defin}

\begin{defin}

A hb-graph is said \textbf{connected} if its support hypergraph is
connected, disconnected otherwise.

\end{defin}

\begin{defin}

A \textbf{connected component} of a hb-graph is a maximal set of vertices
such that every pair of vertices of the component has an m-path in
between them.

\end{defin}

\begin{rmk}

A connected component of a hb-graph is a connected component of one
of its copy hypergraph.

\end{rmk}

\begin{defin}

The \textbf{diameter} of a hb-graph $\mathcal{H}$ - written $\text{diam}\left(\mathcal{H}\right)$
- is defined as:

\[
\text{diam}\left(\mathcal{H}\right)=\underset{x,y\in V}{\max}d(x,y).
\]

\end{defin}

\subsubsection{Adjacency}

\begin{defin}

Let $k$ be a positive integer.

Let consider $k$ vertices not necessarily distinct belonging to $V$.

Let write $V_{k,m}$ the mset consisting of these $k$ vertices with
multiplicity function $m$.

The $k$ vertices are said \textbf{$\boldsymbol{k}$-adjacent} in
$\mathcal{H}$ if it exists $e\in E$ such that $V_{k,m}\subseteq e$.

\end{defin}

Considering a hb-graph $\mathcal{H}$ of m-range $\overline{k}=r_{\mathcal{H}}$,
the hb-graph can't handle more than $\overline{k}$-adjacency in it.
This maximal $k$-adjacency is called the\textbf{ }$\boldsymbol{\overline{k}}$\textbf{-adjacency}
of $\mathcal{H}$.

\begin{defin}

Let consider a hb-edge $e$ in $\mathcal{H}$.

Vertices in the support of $e$ are said \textbf{$\boldsymbol{e^{\star}}$-adjacent}.

Vertices in the hb-edge $e$ with nonzero multiplicity are said \textbf{$\boldsymbol{e}$-adjacent}.

\end{defin}

\begin{rmk}
\begin{itemize}
\item $e^{*}$-adjacency doesn't support redundancy of vertices.
\item $e$-adjacency allows the redundancy of vertices.
\item The only case of equality is where the hb-edge has all its nodes of
multiplicity 1 at the most.
\end{itemize}
\end{rmk}

\begin{defin}

Two hb-edges are said \textbf{incident} if their support intersection
is not empty.

\end{defin}

\subsubsection{Sum of two hb-graphs}

Let $\mathcal{H}_{1}=\left(V_{1},E_{1}\right)$ and $\mathcal{H}_{2}=\left(V_{2},E_{2}\right)$
be two hb-graphs. 

The \textbf{sum of two hb-graphs} $\mathcal{H}_{1}$ and $\mathcal{H}_{2}$
is the hb-graph written $\mathcal{H}_{1}+\mathcal{H}_{2}$ defined
as the hb-graph that has:
\begin{itemize}
\item $V_{1}\cup V_{2}$ as vertex set and where the hb-edges are obtained
from the hb-edges of $E_{1}$ and $E_{2}$ with same multiplicity
for vertices of $V_{1}$ (respectively $V_{2}$) but such that for
each hyperedge in $E_{1}$ (respectively $E_{2}$) the universe is
extended to $V_{1}\cup V_{2}$ and the multiplicity function is extended
such that $\forall v\in V_{2}\backslash V_{1}:\,m\left(v\right)=0$
(respectively $\forall v\in V_{1}\backslash V_{2}:\,m\left(v\right)=0$)
\item $E_{1}+E_{2}$ as hb-edge family, ie the family constituted of the
elements of $E_{1}$ and of the elements of $E_{2}$.
\[
\mathcal{H}_{1}+\mathcal{H}_{2}=\left(V_{1}\cup V_{2},E_{1}+E_{2}\right)
\]
\end{itemize}
This sum is said direct if $E_{1}+E_{2}$ doesn't contain any new
pair of repeated hb-edge than the ones already existing in $E_{1}$
and those already existing in $E_{2}$. In this case the sum is written
$\mathcal{H}_{1}\oplus\mathcal{H}_{2}$.

\subsection{An example}

\begin{example}

Considering $\mathcal{H}=\left(V,E\right)$, with $V=\left\{ v_{1},v_{2},v_{3},v_{4},v_{5},v_{6},v_{7}\right\} $
and $E=\left\{ e_{1},e_{2},e_{3},e_{4}\right\} $ with: $e_{1}=\left\{ v_{1}^{2},v_{4}^{2},v_{5}^{1}\right\} $,
$e_{2}=\left\{ v_{2}^{3},v_{3}^{1}\right\} $, $e_{3}=\left\{ v_{3}^{1},v_{5}^{2}\right\} $,
$e_{4}=\left\{ v_{6}\right\} $.

It holds:

\begin{center}%
\begin{tabular}{|c|c|c|c|c|c|c|}
\hline 
 & $e_{1}$ & $e_{2}$ & $e_{3}$ & $e_{4}$ & $d_{m}\left(v_{i}\right)$ & $\max\left\{ m_{e_{j}}\left(v_{i}\right)\right\} $\tabularnewline
\hline 
$v_{1}$ & 2 & 0 & 0 & 0 & 2 & 2\tabularnewline
\hline 
$v_{2}$ & 0 & 3 & 0 & 0 & 3 & 3\tabularnewline
\hline 
$v_{3}$ & 0 & 1 & 1 & 0 & 2 & 1\tabularnewline
\hline 
$v_{4}$ & 2 & 0 & 0 & 0 & 2 & 1\tabularnewline
\hline 
$v_{5}$ & 1 & 0 & 2 & 0 & 3 & 2\tabularnewline
\hline 
$v_{6}$ & 0 & 0 & 0 & 1 & 1 & 1\tabularnewline
\hline 
$v_{7}$ & 0 & 0 & 0 & 0 & 0 & 0\tabularnewline
\hline 
$\#_{m}e_{j}$ & 5 & 4 & 3 & 1 &  & \tabularnewline
\hline 
\end{tabular}\end{center}

Therefore the order of $\mathcal{H}$ is $O\left(\mathcal{H}\right)=2+3+1+1+2+1+0=10$
and its size is $\left|E\right|=4$.

$v_{7}$ is an isolated vertex.

$e_{1}$ and $e_{3}$ are incident as well as $e_{3}$ and $e_{2}$.
$e_{4}$ is not incident to any hb-edge.

$v_{1}$, $v_{4}$ and $v_{5}$ are $e^{\star}$-adjacent as they
hold in $e_{1}^{*}$.

$v_{1}^{2}$, $v_{4}^{1}$ and $v_{5}^{1}$ are $e$-adjacent as they
hold in $e_{1}$.

The dual of $\mathcal{H}$ is the hb-graph: $\widetilde{\mathcal{H}}=\left(\widetilde{V},\widetilde{E}\right)$
with:
\begin{itemize}
\item $\widetilde{V}=\left\{ \tilde{x_{1}},\tilde{x_{2}},\tilde{x_{3}},\tilde{x_{4}}\right\} $
with $f\left(\tilde{x_{i}}\right)=e_{i}$ for $1\leqslant i\leqslant4$
\item $\tilde{E}=\left\{ \tilde{e_{1}},\tilde{e_{2}},\tilde{e_{3}},\tilde{e_{4}},\tilde{e_{5}},\tilde{e_{6}},\tilde{e_{7}}\right\} $
with:
\begin{itemize}
\item $\tilde{e_{1}}=\tilde{e_{4}}=\left\{ \tilde{x_{1}}^{2}\right\} $
\item $\tilde{e_{2}}=\left\{ \tilde{x_{2}}^{3}\right\} $
\item $\tilde{e_{3}}=\left\{ \tilde{x_{2}}^{1},\tilde{x_{3}}^{1}\right\} $
\item $\tilde{e_{5}}=\left\{ \tilde{x_{1}}^{1},\tilde{x_{3}}^{2}\right\} $
\item $\tilde{e_{6}}=\left\{ \tilde{x_{4}}^{1}\right\} $
\item $\tilde{e_{7}}=\emptyset$
\end{itemize}
\end{itemize}
$\widetilde{\mathcal{H}}$ has duplicated hb-edges and one empty hb-edge.

\end{example}

\section{Algebraic representation of a hb-graph}

\label{sec:algebraic desc}

\subsection{Incidence matrix of a hb-graph}

A mset is well defined by giving itself its universe, its support
and its function of multiplicity. We have seen that a mset can be
represented by a vector called the vector representation of the m-set.

Hb-edges of a given hb-graph have all the same universe.

\begin{defin}

Let $n$ and $p$ be two positive integers.

Let $\mathcal{H}=\left(V,E\right)$ be a non-empty hb-graph, with
vertex set $V=\left\{ v_{i}:i\in\left\llbracket n\right\rrbracket \right\} $
and $E=\left\{ e_{j}:j\in\left\llbracket p\right\rrbracket \right\} $.

The matrix $H=\left[m_{j}\left(v_{i}\right)\right]_{\substack{i\in\left\llbracket n\right\rrbracket \\
j\in\left\llbracket p\right\rrbracket 
}
}$ is called the \textbf{incidence matrix} of the hb-graph $\mathcal{H}$.

\end{defin}

This incidence matrix is intensively used in \citet{ouvrard2018hbgraphdiffusion}
for diffusion by exchanges in hb-graphs.

\subsection{$e$-adjacency tensor of a natural hb-graph}

To build the $e$-adjacency tensor $\mathcal{A}\left(\mathcal{H}\right)$
of a natural hb-graph $\mathcal{H}=\left(V,E\right)$ without repeated
hb-edge - with vertex set $V=\left\{ v_{i}:i\in\left\llbracket n\right\rrbracket \right\} $
and hb-edge set $E=\left\{ e_{j}:j\in\left\llbracket p\right\rrbracket \right\} $
- we use a similar approach that was used in \citet{ouvrard2017cooccurrence}
using the strong link between cubical symmetric tensors and homogeneous
polynomials.

\begin{defin}

An \textbf{elementary hb-graph} is a hb-graph that has only one non
repeated hb-edge in its hb-edge family.

\end{defin}

\begin{claim}

Let $\mathcal{H}=\left(V,E\right)$ be a hb-graph with no repeated
hb-edge.

Then: 
\[
\mathcal{H}=\underset{e\in E}{\bigoplus}\mathcal{H}_{e}
\]

where $\mathcal{H}_{e}=\left(V,\left(e\right)\right)$ is the elementary
hb-graph associated to the hb-edge $e$.

\end{claim}

\begin{proof}

Let $e_{1}\in E$ and $e_{2}\in E$. As $\mathcal{H}$ is with no
repeated hb-edge, $e_{1}+e_{2}$ doesn't contain new pairs of repeated
elements. Thus $\mathcal{H}_{e_{1}}+\mathcal{H}_{e_{2}}$ is a direct
sum.

A straightforward iteration over elements of $e\in E$ leads trivially
to the result.

\end{proof}

We need first to define hypermatrices for the $\overline{k}$-adjacency
of an elementary hb-graph and of a m-uniform hb-graph.

\subsubsection{Normalised $\overline{k}$-adjacency tensor of an elementary hb-graph}

We consider an elementary hb-graph $\mathcal{H}_{e}=\left(V,\left(e\right)\right)$
where $V=\left\{ v_{i}:i\in\left\llbracket n\right\rrbracket \right\} $
and $e$ is a multiset of universe $V$ and multiplicity function
$m$. The support of $e$ is $e^{\star}=\left\{ v_{j_{1}},\ldots,v_{j_{k}}\right\} $
by considering, without loss of generality: $1\leqslant j_{1}<\ldots<j_{k}\leqslant n$
.

$e$ is the multiset: $e=\left\{ v_{j_{1}}^{m_{j_{1}}},\ldots,v_{j_{k}}^{m_{j_{k}}}\right\} $
where $m_{j}=m\left(v_{j}\right)$.

The normalised hypermatrix representation of $e$, written $Q_{e}$,
describes uniquely the m-set $e$. Thus the elementary hb-graph $\mathcal{H}_{e}$
is also uniquely described by $Q_{e}$ as $e$ is the unique hb-edge.
$Q_{e}$ is of rank $r=\#_{m}e=\sum\limits _{j=1}^{k}m_{j}$ and dimension
$n$.

Hence, the definition:

\begin{defin}

Let $\mathcal{H}=\left(V,(e)\right)$ be an elementary hb-graph with
$V=\left\{ v_{i}:i\in\left\llbracket n\right\rrbracket \right\} $
and\textbf{ }$e$ the multiset $\left\{ v_{j_{1}}^{m_{j_{1}}},\ldots,v_{j_{k}}^{m_{j_{k}}}\right\} $
of m-rank $r$, universe $V$ and multiplicity function $m.$

The \textbf{normalised $\overline{k}$-adjacency hypermatrix of an
elementary hb-graph} $\mathcal{H}_{e}$ is the normalised representation
of the multiset $e$: it is the symmetric hypermatrix $Q_{e}=\left(q_{j_{1}\ldots j_{r}}\right)$
of rank $r$ and dimension $n$ where the only nonzero elements are:

\[
q_{\sigma\left(j_{1}\right)^{m_{i\,\sigma\left(j_{1}\right)}}\ldots\sigma\left(j_{k}\right)^{m_{i\,\sigma\left(j_{k}\right)}}}=\dfrac{m_{ij_{1}}!\ldots m_{ij_{k}}!}{\left(r-1\right)!}
\]

where $\sigma\in\mathcal{S}_{\left\llbracket r\right\rrbracket }.$

\end{defin}

In a elementary hb-graph the $\overline{k}$-adjacency corresponds
to $\#_{m}e$-adjacency. This hypermatrix encodes the $\overline{k}$-adjacency
of the elementary hb-graph; as the $\overline{k}$-adjacency corresponds
to $e$-adjacency in such a hb-graph is encodes also the $e$-adjacency
of the elementary hb-graph.

\subsubsection{hb-graph polynomial}

\paragraph*{Homogeneous polynomial associated to a hypermatrix: }

With a similar approach than in \citet{ouvrard2017cooccurrence} where
full details are given, let write $\boldsymbol{e_{1}},\ldots,\boldsymbol{e_{n}}$
the canonical basis of $\mathbb{R}^{n}$.

$\left(\boldsymbol{e_{i_{1}}}\otimes\ldots\otimes\boldsymbol{e_{i_{k}}}\right)_{i_{1},\ldots,i_{k}\in\left\llbracket n\right\rrbracket }$
is a basis of $\mathcal{L}_{k}^{0}\left(\mathbb{K}^{n}\right)$, where
$\otimes$ is the Segre outerproduct. 

A tensor $\boldsymbol{Q}\in\mathcal{L}_{k}^{0}\left(\mathbb{K}^{n}\right)$
is associated to an hypermatrix $Q=\left(q_{\,i_{1}\ldots i_{r}}\right)_{i_{1},\ldots,i_{r}\in\left\llbracket n\right\rrbracket }$
by writting $\boldsymbol{Q}$ as:

\[
\boldsymbol{Q}=\sum\limits _{i_{1},\ldots,i_{r}\in\left\llbracket n\right\rrbracket }q_{\,i_{1}\ldots i_{r}}\boldsymbol{e_{i_{1}}}\otimes\ldots\otimes\boldsymbol{e_{i_{r}}}
\]

Considering $n$ variables $z_{i}$ attached to the $n$ vertices
$v_{i}$ and $z=\sum\limits _{i\in\left\llbracket n\right\rrbracket }z_{i}\boldsymbol{e_{i}}$,
the multilinear matrix product $\left(z,\ldots,z\right).Q=\left(z\right)_{\left[r\right]}.Q$
is a polynomial $P\left(\boldsymbol{z_{0}}\right)$\footnote{As a reminder:$\boldsymbol{z_{0}}=\left(z_{1},...,z_{n}\right)$}:

\[
P\left(\boldsymbol{z_{0}}\right)=\sum\limits _{i_{1},\ldots,i_{r}\in\left\llbracket n\right\rrbracket }q_{\,i_{1}\ldots i_{r}}z_{i_{1}}\ldots z_{i_{r}}
\]
 of degree $r$.

\subsubsection*{Elementary hb-graph polynomial:}

Considering a hb-graph $\mathcal{H}_{e}=\left(V,\left(e\right)\right)$
with $V=\left\{ v_{i}:i\in\left\llbracket n\right\rrbracket \right\} $
and\textbf{ }$e$ the multiset $\left\{ v_{j_{1}}^{m_{j_{1}}},\ldots,v_{j_{k}}^{m_{j_{k}}}\right\} $
of m-rank $r$, universe $V$ and multiplicity function $m$.

Using the normalised $\overline{k}$-adjacency hypermatrix $Q_{e}=\left(q_{\,i_{1}\ldots i_{r}}\right)_{i_{1},\ldots,i_{r}\in\left\llbracket n\right\rrbracket }$,
which is symmetric, we can write the reduced version of its attached
homogeneous polynomial $P_{e}$: 
\[
P_{e}\left(\boldsymbol{z_{0}}\right)=\dfrac{r!}{m_{j_{1}}!\ldots m_{j_{k}}!}q_{j_{1}^{m_{j_{1}}}\ldots j_{k}^{m_{j_{k}}}}z_{j_{1}}^{m_{j_{1}}}\ldots z_{j_{k}}^{m_{j_{k}}}.
\]

\paragraph*{Hb-graph polynomial: }

Considering a hb-graph $\mathcal{H}=\left(V,E\right)$ with no-repeated
hb-edge, with $V=\left\{ v_{i}:i\in\left\llbracket n\right\rrbracket \right\} $
and $E=\left(e_{i}\right)_{i\in\left\llbracket p\right\rrbracket }.$

This hb-graph can be summarized by a polynomial of degree $r_{\mathcal{H}}=\underset{e\in E}{\max}\#_{m}\left(e\right)$:

\begin{eqnarray*}
P\left(\boldsymbol{z_{0}}\right) & = & \sum\limits _{i\in\left\llbracket p\right\rrbracket }c_{e_{i}}P_{e_{i}}\left(\boldsymbol{z_{0}}\right)\\
 & = & \sum\limits _{i\in\left\llbracket p\right\rrbracket }c_{e_{i}}\dfrac{r_{i}!}{m_{ij_{1}}!\ldots m_{ij_{k_{i}}}!}q_{j_{1}^{m_{i\,j_{1}}}\ldots j_{k_{i}}^{m_{i\,j_{k_{i}}}}}z_{j_{1}}^{m_{ij_{1}}}\ldots z_{j_{k_{i}}}^{m_{ij_{k_{i}}}}
\end{eqnarray*}

where $c_{e_{i}}$ is a technical coefficient. $P\left(z_{0}\right)$
is called the \textbf{hb-graph polynomial}. The choice of $c_{e_{i}}$
is made in order to retrieve the m-degree of the vertices from the
$e$-adjacency tensor.

\subsubsection{$\overline{k}$-adjacency hypermatrix of a m-uniform natural hb-graph}

We now extend to m-uniform hb-graph the $\overline{k}$-adjacency
hypermatrix obtained in the case of an elementary hb-graph.

In the case of a $r$-m-uniform natural hb-graph with no repeated
hb-edge, each hb-edge has the same $m$-cardinality $r$. Hence the
$\overline{k}$-adjacency of a $r$-m-uniform hb-graph corresponds
to $r$-adjacency where $r$ is the m-rank of the hb-graph. The $\overline{k}$-adjacency
tensor of the hb-graph has rank $r$ and dimension $n$. The elements
of the $\overline{k}$-adjacency hypermatrix are: 
\[
a_{i_{1}\ldots i_{r}}
\]
 with $i_{1},\ldots,i_{r}\in\left\llbracket n\right\rrbracket $. 

The associated hb-graph polynomial is homogeneous of degree $r.$

We obtain the definition of the $\overline{k}$-adjacency tensor of
a $r$-m-uniform hb-graph by summing the $\overline{k}$-adjacency
tensor attached to each hyperedge with a coefficient $c_{i}$ equals
to 1 for each hyperedge.

\begin{defin}

Let $\mathcal{H}=\left(V,E\right)$ be a hb-graph. $V=\left\{ v_{i}:i\in\left\llbracket n\right\rrbracket \right\} $.

The \textbf{$\overline{k}$-adjacency hypermatrix of a $r$-m-uniform
hb-graph $\mathcal{H}=\left(V,E\right)$} is the hypermatrix $A_{\mathcal{H}}=\left(a_{i_{1}\ldots i_{r}}\right)_{i_{1},\ldots,i_{r}\in\left\llbracket n\right\rrbracket }$
defined by:

\[
A_{\mathcal{H}}=\sum\limits _{i\in\left\llbracket p\right\rrbracket }Q_{e_{i}}
\]

where $Q_{e_{i}}$ is the $\overline{k}$-adjacency hypermatrix of
the elementary hb-graph associated to the hb-edge $e_{i}=\left\{ v_{j_{1}}^{m_{ij_{1}}},\ldots,v_{j_{k_{i}}}^{m_{ij_{k_{i}}}}\right\} \in E$.

The only non-zero elements of $Q_{e_{i}}$ are the elements of indices
obtained by permutation of the multiset $\left\{ j_{1}^{m_{ij_{1}}},\ldots,j_{k_{i}}^{m_{ij_{k_{i}}}}\right\} $
and are all equals to $\dfrac{m_{ij_{1}}!\ldots m_{ij_{k_{i}}}!}{\left(r-1\right)!}$.

\end{defin}

\begin{rmk}

When a $r$-m-uniform hb-graph has 1 as vertex multiplicity for any
vertices in each hb-edge support of all hb-edges, then this hb-graph
is a $r$-uniform hypergraph: in this case, we retrieve the result
of the degree-normalized tensor defined in \citet{cooper2012spectra}.

\end{rmk}

\begin{claim}

The $m$-degree of a vertex $v_{j}$ in a $r$-m-uniform hb-graph
$\mathcal{H}$ of $\overline{k}$-adjacency hypermatrix is:
\[
\deg_{m}\left(v_{j}\right)=\sum\limits _{j_{2},...,j_{r}\in\left\llbracket n\right\rrbracket }a_{jj_{2}...j_{r}}.
\]

\end{claim}

\begin{proof}

$\sum\limits _{j_{2},...,j_{r}\in\left\llbracket n\right\rrbracket }a_{jj_{2}\ldots j_{r}}$
has non-zero terms only for corresponding hb-edges $e_{i}$ that have
$v_{j}$ in it. For such a hb-edge containing $v_{j}$, it is described
by $e_{i}=\left\{ v_{j}^{m_{i\,j}},v_{l_{2}}^{m_{i\,l_{2}}},\ldots,v_{l_{k_{i}}}^{m_{i\,l_{k_{i}}}}\right\} $.
It means that the multiset $\left\{ \left\{ j_{2},\ldots,j_{r}\right\} \right\} $
corresponds exactly to the multiset $\left\{ j^{m_{i\,j-1}},l_{2}^{m_{i\,l_{2}}},\ldots,l_{k_{i}}^{m_{i\,l_{k_{i}}}}\right\} .$
For each $e_{i}$ such that $v_{j}\in e_{i},$ there is $\dfrac{\left(r-1\right)!}{\left(m_{i\,j}-1\right)!m_{i\,l_{2}}!\ldots m_{i\,l_{k_{i}}}!}$
possible permutation of the indices $j_{2}$ to $j_{l}$ and $a_{jj_{2}...j_{r}}=\dfrac{m_{i\,j}!m_{i\,l_{2}}!\ldots m_{i\,l_{k_{i}}}!}{\left(r-1\right)!}$. 

Also: $\sum\limits _{j_{2},...,j_{r}\in\left\llbracket n\right\rrbracket }a_{jj_{2}\ldots j_{r}}=\sum\limits _{i\in\left\llbracket p\right\rrbracket \,:\,v_{j}\in e_{i}}m_{i\,j}=\text{deg}_{m}\left(v_{j}\right)$.

\end{proof}

\subsubsection{Elementary operations on hb-graphs}

In \citet{ouvrard2017cooccurrence}, we describe two elementary operations
that are used in the hypergraph uniformisation process. We describe
here two similar operations and some additional operations for hb-graphs.

\textbf{\vspace{1em}\hspace{-1.25em}}%
\noindent\fbox{\begin{minipage}[t]{1\columnwidth - 2\fboxsep - 2\fboxrule}%
\begin{operation}

Let $\mathcal{H}=\left(V,E\right)$ be a hb-graph.

Let $w_{1}$ be a constant weighted function on hb-edges with constant
value $1$.

The weighted hb-graph $\mathcal{H}_{1}=\left(V,E,w_{1}\right)$ is
called the \textbf{canonical weighted hb-graph} of $\mathcal{H}.$

The application $\phi_{\text{cw}}:\mathcal{H}\mapsto\mathcal{H}_{1}$
is called the \textbf{canonical weighting operation}.

\end{operation}%
\end{minipage}}

\textbf{\vspace{1em}\hspace{-1.25em}}%
\noindent\fbox{\begin{minipage}[t]{1\columnwidth - 2\fboxsep - 2\fboxrule}%
\begin{operation}

Let $\mathcal{H}=\left(V,E,w_{1}\right)$ be a canonical weighted
hb-graph.

Let $c\in\mathbb{R}^{++}$. Let $w_{c}$ be a constant weighted function
on hb-edges with constant value $c$.

The weighted hb-graph $\mathcal{H}_{c}=\left(V,E,w_{c}\right)$ is
called the \textbf{$c$-dilatated hb-graph} of $\mathcal{H}.$

The application $\phi_{\text{c-d}}:\mathcal{H}_{1}\mapsto\mathcal{H}_{c}$
is called the \textbf{c-dilatation operation}.

\end{operation}%
\end{minipage}}

\textbf{\vspace{1em}\hspace{-1.25em}}%
\noindent\fbox{\begin{minipage}[t]{1\columnwidth - 2\fboxsep - 2\fboxrule}%
\begin{operation}

Let $\mathcal{H}_{w}=\left(V,E,w\right)$ be a weighted hb-graph.
Let $y\notin V$ be a new vertex.

The \textbf{y-complemented hbgraph} of $\mathcal{H}_{w}$ is the hbgraph
$\mathcal{\tilde{H}}_{\tilde{w}}=\left(\tilde{V},\tilde{E},\tilde{w}\right)$
where $\tilde{V}=V\cup\left\{ y\right\} $, $\tilde{E}=\left(\xi\left(e\right)\right)_{e\in E}$
- with the map $\xi:E\rightarrow\mathcal{M}\left(\tilde{V}\right)$
such that for all $e\in E$, $\xi\left(e\right)\in\mathcal{M}\left(\tilde{V}\right)$
and is the multiset $\left\{ x^{m_{\xi(e)}(x)}:x\in\tilde{V}\right\} $,
with $m_{\xi(e)}(x)=\begin{cases}
m_{e}(x) & \text{if }x\in e^{\star}\\
r_{\mathcal{H}}-\#_{m}e & \text{if }x=y
\end{cases}$ - and, the weight function is $\tilde{w}$ is such that $\forall e\in E$:
$\tilde{w}\left(\xi(e)\right)=w(e).$

The application $\phi_{\text{y-c}}:\mathcal{H}_{w}\mapsto\mathcal{\tilde{H}}_{\tilde{w}}$
is called the \textbf{y-complemented operation}.

\end{operation}%
\end{minipage}}

\textbf{\vspace{1em}\hspace{-1.25em}}%
\noindent\fbox{\begin{minipage}[t]{1\columnwidth - 2\fboxsep - 2\fboxrule}%
\begin{operation}

Let $\mathcal{H}_{w}=\left(V,E,w\right)$ be a weighted hb-graph.
Let $y\notin V$ be a new vertex. Let $\alpha\in\mathbb{R}^{++}.$

The \textbf{$y^{\alpha}$-vertex-increased hbgraph} of $\mathcal{H}_{w}$
is the hbgraph $\mathcal{H}_{w^{+}}^{+}=\left(V^{+},E^{+},w^{+}\right)$
where $V^{+}=V\cup\left\{ y\right\} $, $E^{+}=\left(\phi\left(e\right)\right)_{e\in E}$
- with the map $\phi:E\rightarrow\mathcal{M}\left(V^{+}\right)$ such
that for all $e\in E$, $\phi\left(e\right)\in\mathcal{M}\left(V^{+}\right)$
and is the multiset $\left\{ x^{m_{\phi(e)}(x)}:x\in V^{+}\right\} $,
with $m_{\phi(e)}(x)=\begin{cases}
m_{e}(x) & \text{if }x\in e^{\star}\\
\alpha & \text{if }x=y
\end{cases}$ - and, the weight function is $w^{+}$ is such that $\forall e\in E$:
$w^{+}\left(\phi(e)\right)=w(e).$

The application $\phi_{y^{\alpha}\text{-v}}:\mathcal{H}_{w}\mapsto\mathcal{H}_{w^{+}}^{+}$
is called the \textbf{$y^{\alpha}$-vertex-increasing operation}.

\end{operation}%
\end{minipage}}

\textbf{\vspace{1em}\hspace{-1.25em}}%
\noindent\fbox{\begin{minipage}[t]{1\columnwidth - 2\fboxsep - 2\fboxrule}%
\begin{operation}

The \textbf{merged hb-graph} $\widehat{\mathcal{H}_{\widehat{w}}}=\left(\widehat{V},\widehat{E},\widehat{w}\right)$
of a family $\left(\mathcal{H}_{i}\right)_{i\in I}$ of weighted hb-graphs
with $\forall i\in I:\,\mathcal{H}_{i}=\left(V_{i},E_{i},w_{i}\right)$
is the weighted hb-graph with vertex set $\widehat{V}=\bigcup\limits _{i\in I}V_{i}$,
with hb-edge family $\widehat{E}=\left(\psi\left(e\right)\right)_{e\in\sum\limits _{i\in I}E_{i}}$\footnote{$\sum\limits _{i\in I}E_{i}$ is the family obtained with all elements
of each family $E_{i}.$} - with the map $\psi:\sum\limits _{i\in I}E_{i}\rightarrow\mathcal{M}\left(\widehat{V}\right)$
such that for all $e\in\sum\limits _{i\in I}E_{i}$, $\psi\left(e\right)\in\mathcal{M}\left(\widehat{V}\right)$
and is the multiset $\left\{ x^{m_{\psi(e)}(x)}:x\in\widehat{V}\right\} $,
with $m_{\psi(e)}(x)=\begin{cases}
m_{e}(x) & \text{if }x\in e^{\star}\\
0 & \text{otherwise}
\end{cases}$- and, such that $\forall e\in E_{i}$, $\widehat{w}(e)=w_{i}(e).$

The application $\phi_{\text{m}}:\left(\mathcal{H}_{i}\right)_{i\in I}\mapsto\widehat{\mathcal{H}}$
is called the \textbf{merging operation}.

\end{operation}%
\end{minipage}}

\textbf{\vspace{1em}\hspace{-1.25em}}%
\noindent\fbox{\begin{minipage}[t]{1\columnwidth - 2\fboxsep - 2\fboxrule}%
\begin{operation}

Decomposing a hb-graph $\mathcal{H}=\left(V,E\right)$ into a family
of hb-graphs $\left(\mathcal{H}_{i}\right)_{i\in I}$, where $\mathcal{H}_{i}=\left(V,E_{i}\right)$
such that $\mathcal{H}=\bigoplus\limits _{i\in I}\mathcal{H}_{i}$
is called the \textbf{decomposition operation} $\phi_{\text{d}}:\mathcal{H}\mapsto\left(\mathcal{H}_{i}\right)_{i\in I}$.

\end{operation}%
\end{minipage}}

\begin{rmk}

The direct sum of two hb-graphs appears as a merging operation of
the two hb-graphs.

\end{rmk}

\begin{defin}

Let $\mathcal{H}=\left(V,E\right)$ and $\mathcal{H}'=\left(V',E'\right)$
be two hb-graphs.

Let $\mathcal{\phi}:\mathcal{H}\mapsto\mathcal{H}'$. 

$\phi$ is said \textbf{preserving $\boldsymbol{e}$-adjacency} if
vertices of $V'$ that are $e$-adjacent in $\mathcal{H}'$ are either
$e$-adjacent vertices in $\mathcal{H}$ or the maximal subset of
these vertices that are in $V$ are $e$-adjacent in $\mathcal{H}$.

$\phi$ is said \textbf{preserving exactly $\boldsymbol{e}$-adjacency}
if vertices that are $e$-adjacent in $\mathcal{H}'$ are $e$-adjacent
in $\mathcal{H}$ and reciprocally.

\end{defin}

We can extend these definitions to $\psi:\left(\mathcal{H}_{i}\right)_{i\in I}\mapsto\mathcal{H}'$.

\begin{defin}

Let $\left(\mathcal{H}_{i}\right)_{i\in I}$ be a family of hb-graphs
with $\forall i\in I,\,\mathcal{H}_{i}=\left(V_{i},E_{i}\right)$
and $\mathcal{H}'=\left(V',E'\right)$ a hb-graph.

Let consider $\psi:\left(\mathcal{H}_{i}\right)_{i\in I}\mapsto\mathcal{H}'$. 

$\psi$ is said \textbf{preserving }$e$-adj\textbf{acency} if vertices
that are $e$-adjacent in $\mathcal{H}'$ are either $e$-adjacent
vertices in exactly one of the $\mathcal{H}_{i},i\in I$ or the maximal
subset of these vertices that is in $V=\bigcup\limits _{i\in I}V_{i}$
is $e$-adjacent in exactly one of the $\mathcal{H}_{i}.$

$\phi$ is said \textbf{preserving exactly }$e$-adj\textbf{acency}
if vertices that are $e$-adjacent in $\mathcal{H}'$ are $e$-adjacent
in exactly one of the $\mathcal{H}_{i},i\in I$ and reciprocally.

\end{defin}

We can extend these definitions to $\nu:\mathcal{H}\mapsto\left(\mathcal{H}_{i}\right)_{i\in I}.$

\begin{defin}

Let $\left(\mathcal{H}_{i}\right)_{i\in I}$ be a family of hb-graphs
with $\forall i\in I,\,\mathcal{H}_{i}=\left(V_{i},E_{i}\right)$
and $\mathcal{H}=\left(V,E\right)$ a hb-graph.

Let consider $\nu:\mathcal{H}\mapsto\left(\mathcal{H}_{i}\right)_{i\in I}.$

$\psi$ is said \textbf{preserving }$e$-adj\textbf{acency} if vertices
that are $e$-adjacent in one of the $\mathcal{H}_{i},i\in I$ are
either $e$-adjacent vertices in $\mathcal{H}$ or the maximal subset
of these vertices that is in $V$ is $e$-adjacent in $\mathcal{H}.$

$\phi$ is said \textbf{preserving exactly }$e$-adj\textbf{acency}
if vertices that are $e$-adjacent in one of the $\mathcal{H}_{i},i\in I$
are $e$-adjacent in $\mathcal{H}$ and reciprocally.

\end{defin}

\begin{claim}

Let $\mathcal{H}=\left(V,E\right)$ be a hb-graph.

The canonical weighting operation, the $c$-dilatation operation,
the merging operation and, the decomposition operation preserve exactly
$e$-adjacency.

The $y$-complemented operation and the $y^{\alpha}$-vertex-increasing
operation preserve $e$-adjacency.

\end{claim}

\begin{proof}

Immediate.

\end{proof}

\begin{claim}

The composition of two operations which preserve (respectively exactly)
$e$-adjacency preserves (respectively exactly) $e$-adjacency.

The composition of two operations where one preserves exactly $e$-adjacency
and the other preverves $e$-adjacency preserves $e$-adjacency.

\end{claim}

\begin{proof}

Immediate.

\end{proof}

\subsubsection{Processes involved for building the $e$-adjacency tensor}

In a general natural hb-graph $\mathcal{H}$, hb-edges are not forced
to have same m-cardinality: the rank of the $\overline{k}$-adjacency
tensor of the elementary hb-graph associated to each hb-edge depends
on the m-cardinality of the hb-edge. As a consequence the hb-graph
polynomial is no more homogeneous. Nonetheless techniques to homogenize
such a polynomial are well known.

The hb-graph m-uniformisation process (Hm-UP) tranform a given hb-graph
of m-range $r_{\mathcal{H}}$ into a $r_{\mathcal{H}}$-m-uniform
hb-graph written $\mathcal{\overline{H}}$: this uniformisation can
be mapped to the homogenization of the attached polynomial of the
original hb-graph, called the polynomial homogenization process (PHP). 

The Hm-UP can be achieved by different means: 
\begin{itemize}
\item \textbf{straightforward m-uniformisation} levels directly all hb-edges
by adding a Null vertex $y_{0}$ with multiplicity the difference
between the hb-graph m-rank and the hb-edge m-cardinality. It is achieved
by considering the $y_{0}$-complemented hb-graph of $\mathcal{H}.$
\item \textbf{silo m-uniformisation} processes the hb-edges of given m-cardinality
$r$, regrouped in the dilatated hbgraph $\mathcal{H}_{c_{r},r}$
- obtained by canonical weighting and dilatation of the $r$-m-uniform
sub-hb-graph $\mathcal{H}_{r}$ of $\mathcal{H}$ containing all its
hb-edges of m-cardinality $r$ - by adding a $r$-dependent null vertex
$y_{r}$ in multiplicity $r_{\mathcal{H}}-r$ so it levels the hb-edges
to a constant m-cardinality of $r_{\mathcal{H}}$; it uses the $y_{r}^{r_{\mathcal{H}}-r}$-vertex-increasing
operation on each $\mathcal{H}_{c_{r},r}$. A merging operation is
then used to gather all the $y_{r}^{r_{\mathcal{H}}-r}$-vertex-increased
hb-graph into a single $r_{\mathcal{H}}$-m-uniform.
\item \textbf{layered m-uniformisation} processes m-uniform hb-subgraphs
of increasing m-cardinality by successively adding a vertex and merging
it to the hb-subgraph of the above layer. The layered homogenisation
process applied to hypergraphs was explained with full details in
\citet{ouvrard2017cooccurrence}; it involves a two-phase step iterations
based on successive $\left\{ y_{k}^{1}\right\} $-vertex-increased
hb-graphs and merging with the dilatated weighted hb-graph of the
next layer.
\end{itemize}

\subsubsection{On the choice of the technical coefficient $c_{e_{i}}$}

\label{subsec:Choice_technical}

The technical coefficient $c_{e_{i}}$ has to be chosen such that
by using the elements of the $e$-adjacency hypermatrix $A=\left(a_{i_{1}\ldots i_{r}}\right)_{i_{1},\ldots,i_{r}\in\left\llbracket n\right\rrbracket }$
, the hypermatrix allows to retrieve:

1. the m-degree of the vertices: $\sum\limits _{i_{2},\ldots,i_{r}\in\left\llbracket n\right\rrbracket }a_{ii_{2}\ldots i_{r}}=\deg_{m}\left(v_{i}\right)$.

2. the number of hb-edges $\left|E\right|$ .

The approach is similar to the one used in Part 1: we consider a hb-graph
$\mathcal{H}$ that we decompose in a family of $r$-m-uniform hb-graphs
$\left(\mathcal{H}_{r}\right)_{r\in\left\llbracket r_{\mathcal{H}}\right\rrbracket }$.

To achieve it, let $\mathcal{R}$ be the equivalency relation defined
on $E$ the family of hb-edges of $\mathcal{H}$: $e\mathcal{R}e'\Leftrightarrow\#_{m}e=\#_{m}e'$.

$E/\mathcal{R}$ is the set of classes of hb-edges of same m-cardinality.
The elements of $E/\mathcal{R}$ are the sets: $E_{r}=\left\{ e\in E:\,\#_{m}e=r\right\} $.

Considering $R=\left\{ r:E_{r}\in E/\mathcal{R}\right\} $, it is
set $E_{r}=\emptyset$ for all $r\in\left\llbracket r_{\mathcal{H}}\right\rrbracket \backslash R$.

Let consider the hb-graphs: $\mathcal{H}_{k}=\left(V,E_{r}\right)$
for all $r\in\left\llbracket r_{\mathcal{H}}\right\rrbracket $ which
are all $r$-m-uniform.

It holds: $E=\bigcup\limits _{r\in\left\llbracket r_{\mathcal{H}}\right\rrbracket }E_{r}$
and $E_{r_{1}}\cap E_{r_{2}}=\emptyset$ for all $r_{1}\neq r_{2}$,
hence $\left(E_{r}\right)_{r\in\left\llbracket r_{\mathcal{H}}\right\rrbracket }$
constitutes a partition of $E$ which is unique from the way it has
been defined.

Hence: 
\[
\mathcal{H}=\bigoplus\limits _{r\in\left\llbracket r_{\mathcal{H}}\right\rrbracket }\mathcal{H}_{r}.
\]

Each of these $r$-m-uniform hb-graph $\mathcal{H}_{r}$ can be associated
to a $\overline{k}$-adjacency tensor $\mathcal{A}_{r}$ viewed as
a hypermatrix $A_{\mathcal{H}_{r}}=\left(a_{(r)i_{1}\ldots i_{r}}\right)$
of order $r$, hypercubic and, symmetric of dimension $\left|V\right|=n$. 

We write $\left(a_{i_{1}\ldots i_{r_{\mathcal{H}}}}\right)_{i_{1},\ldots,i_{r_{\mathcal{H}}}\in\left\llbracket n_{1}\right\rrbracket }$
the $e$-adjacency hypermatrix associated to $\mathcal{H}$: it is
going to be built hereafter and $n_{1}$ depends on the way the hypermatrix
is built.

The number of hb-edges in $\mathcal{H}_{r}$ is given by summing the
elements of $\mathcal{A}_{r}$:

\begin{eqnarray*}
\sum\limits _{i_{1},\ldots,i_{r}\in\left\llbracket n\right\rrbracket }a_{(r)i_{1}\ldots i_{r}} & = & \sum\limits _{i=1}^{n}\sum\limits _{i_{2},\ldots,i_{r}\in\left\llbracket n\right\rrbracket }a_{(r)ii_{2}\ldots i_{r}}\\
 & = & \sum\limits _{i=1}^{n}\text{deg}_{m}\left(v_{i}\right)\\
 & = & r\left|E_{r}\right|
\end{eqnarray*}

In the whole hb-graph $\mathcal{H}$ can also be calculated as:

\begin{eqnarray*}
\sum\limits _{i_{1},\ldots,i_{r_{\mathcal{H}}}\in\left\llbracket n_{1}\right\rrbracket }a_{i_{1}\ldots i_{r_{\mathcal{H}}}} & = & r_{\mathcal{H}}\left|E\right|.
\end{eqnarray*}

As 
\begin{eqnarray*}
\left|E\right| & = & \sum\limits _{r=1}^{n}\left|E_{r}\right|\\
 & = & \sum\limits _{r=1}^{n}\dfrac{1}{r}\sum\limits _{i=1}^{n}\text{deg}_{m}\left(v_{i}\right)\\
 & = & \sum\limits _{r=1}^{n}\dfrac{1}{r}\sum\limits _{i_{1},\ldots,i_{r}\in\left\llbracket n\right\rrbracket }a_{(r)i_{1}\ldots i_{r}}
\end{eqnarray*}

Hence, it follows:

\begin{eqnarray*}
\sum\limits _{i_{1},\ldots,i_{r_{\mathcal{H}}}\in\left\llbracket n_{1}\right\rrbracket }a_{i_{1}\ldots i_{r_{\mathcal{H}}}} & = & \sum\limits _{r=1}^{n}\dfrac{r_{\mathcal{H}}}{r}\sum\limits _{i_{1},\ldots,i_{r}\in\left\llbracket n\right\rrbracket }a_{(r)i_{1}\ldots i_{r}}.
\end{eqnarray*}

Also, choosing for all $i\in\left\llbracket p\right\rrbracket $ such
that $\#_{m}e_{i}=r$, $c_{e_{i}}=\dfrac{r_{\mathcal{H}}}{r}$.

We write for all $r\in\left\llbracket r_{\mathcal{H}}\right\rrbracket $:
$c_{r}=\dfrac{r_{\mathcal{H}}}{r}$. It is the technical coefficient
for the corresponding layer of level $r$ of the hb-graph $\mathcal{H}$.

Hence, the HUP is initiated by applying the canonical weighting to
each m-uniform hb-graph $\mathcal{H}_{r}$ that transforms it into
$\mathcal{H}_{r,1}$. Then the $c_{r}$-dilatation operation is applied
to each weighted m-uniform hb-graph $\mathcal{H}_{r,1}$ to obtain
its $c_{r}$-dilatated hb-graph $\mathcal{H}_{r,c_{r}}$.

\subsubsection{Straightforward approach}

\paragraph*{Straightforward m-uniformisation:}

\begin{figure}
\begin{center}
\begin{tikzpicture}[->,>=stealth',scale=1, every node/.append style={transform shape}]
\node[state=green, minimum width=1cm, minimum height=1cm] (H) {$\mathcal{H}$};
\node[
		minimum width=0.1cm, 
		minimum height=1cm,
		right of=H,
		node distance=1.5cm] (bigBrackLeft) {$\Bigg(\setstackgap{S}{2pt}\raisebox{-1.2em}{\Shortstack{. . . . . . . . . .}}$};
\node[state=green, 
		minimum width=1cm, 
		minimum height=1cm,
		right of=bigBrackLeft,
		node distance=0.7cm] (Hr) {$\mathcal{H}_r$};
\node[right of=H,
		yshift=0.25cm,
		node distance=1cm] (phid) {$\phi_d$};
\node[state=green, 
		minimum width=1cm, 
		minimum height=1cm,
		right of=Hr,
		node distance=2cm] (Hr1) {$\mathcal{H}_{r,1}$};
\node[right of=Hr,
		yshift=0.25cm,
		node distance=1cm] (phicw) {$\phi_\textrm{cw}$};
\node[right of=Hr1,
		yshift=0.25cm,
		node distance=1cm] (phicd) {$\phi_\textrm{c-d}$};
\node[state=green, 
		minimum width=1cm, 
		minimum height=1cm,
		right of=Hr1,
		node distance=2cm] (Hrcr) {$\mathcal{H}_{r,c_r}$};
\node[state=orange!50!blue, 
		minimum width=4.6cm, 
		minimum height=2cm,
		right of=Hrcr,
		node distance=2.9cm] (spec) {};
\node[right of=Hrcr,
		yshift=0.75cm,
		node distance=1.75cm] (spec_text) {specific};
\node[
		minimum width=0.1cm, 
		minimum height=1cm,
		right of=Hrcr,
		node distance=1.2cm] (bigBrackRight) {$\setstackgap{S}{2pt}\raisebox{-1.2em}{\Shortstack{. . . . . . . . . .}}\Bigg)_{r \in r_\mathcal{H}}$};
\node[right of=bigBrackRight,
		yshift=0.25cm,
		node distance=0.25cm] (phim) {$\phi_\textrm{m}$};
\node[state=green, 
		minimum width=1cm, 
		minimum height=1cm,
		right of=bigBrackRight,
		node distance=1.25cm] (Hwd) {$\mathcal{H}_{w,d}$};

\node[right of=Hwd,
		yshift=0.25cm,
		node distance=1cm] (phiyc) {$\phi_{y_1\textrm{-c}}$};
\node[state=green, 
		minimum width=1cm, 
		minimum height=1cm,
		right of=Hwd,
		node distance=2cm] (Htwd) {$\tilde{\mathcal{H}}_\textrm{w,d}$};

\path (H) edge ([xshift=-0.1cm]bigBrackLeft.center)
(Hr) edge (Hr1)
(Hr1) edge (Hrcr)
([xshift=-0.2cm]bigBrackRight.center) edge (Hwd)
(Hwd) edge (Htwd);

\end{tikzpicture}
\end{center}

\caption{Operations on the original hb-graph to m-uniformize it in the straightforward
approach. Parenthesis with vertical dots indicate parallel operations.}
\label{Fig: Operations uniformisation straightforward}
\end{figure}
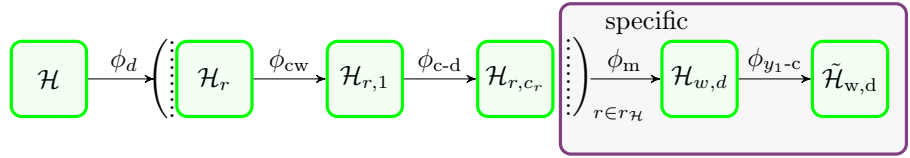

We first decompose $\mathcal{H}=\bigoplus\limits _{r\in\left\llbracket r_{\mathcal{H}}\right\rrbracket }\mathcal{H}_{r}$
as viewed in sub-section \ref{subsec:Choice_technical}.

We transform each $\mathcal{H}_{r},r\in\left\llbracket r_{\mathcal{H}}\right\rrbracket $
into a canonical weighted hb-graph $\mathcal{H}_{r,1}$ that we dilate
using the dilatation coefficient obtaining the $c_{r}$-dilatated
hb-graph $\mathcal{H}_{r,c_{r}}.$

This family $\left(\mathcal{H}_{r,c_{r}}\right)$ is then merged into
the hb-graph: $\mathcal{H}_{w,d}=\bigoplus\limits _{r\in\left\llbracket r_{\mathcal{H}}\right\rrbracket }\mathcal{H}_{r,c_{r}}$.
To get a m-uniform hb-graph at last we generate a vertex $y_{1}\notin V$
and apply to $\mathcal{H}_{w,d}$ the $y_{1}$-complemented operation
to obtain $\tilde{\mathcal{H}}_{w,d}$ the $y_{1}$-complemented hb-graph
of $\mathcal{H}_{w,d}$.

The different steps are summarized in Figure \ref{Fig: Operations uniformisation straightforward}.

\begin{claim}The transformation $\phi_{s}:\mathcal{H}\mapsto\tilde{\mathcal{H}}_{w,d}$
preserves the $e$-adjacency.

\end{claim}

\begin{proof}
$\phi_{\text{s}}=\phi_{y_{1}\text{-c}}\circ\phi_{\text{m}}\circ\left(
\setstackgap{S}{2pt}\raisebox{-1.2em}{\Shortstack{. . . . . . . . . .}}\phi_{\text{c-d}}\circ\phi_{\text{cw}}\raisebox{-1.2em}{\Shortstack{. . . . . . . . . .}}\right)\circ\phi_{\text{d}}.$\footnote{$\left(
\setstackgap{S}{2pt}\raisebox{-1.2em}{\Shortstack{. . . . . . . . . .}}\hspace{2em}\raisebox{-1.2em}{\Shortstack{. . . . . . . . . .}}\right)_{\textrm{... }\in\textrm{ ...}}$ indicates parallel operations on each member of the family indicated
in index of the right parenthesis.}

The operations involved either preserve $e$-adjacency or preserve
exactly $e$-adjacency, also by composition $\phi_{\text{s}}$ preserve
$e$-adjacency.

\end{proof}

\paragraph*{Straightforward homogenization:}

To homogenize the hb-graph, we add an additional vertex $N$ into
the universe of the hb-graph, ie the vertex set, corresponding to
one additional variable $y_{1}$.

The $\overline{k}$-adjacency hypermatrix of the hb-edge $e_{i}=\left\{ v_{j_{1}}^{m_{ij_{1}}},\ldots,v_{j_{k_{i}}}^{m_{ij_{k_{i}}}}\right\} $
is $Q_{e_{i}}$ of rank $\rho_{i}=\#_{m}e_{i}$ and dimension $n$.
The corresponding reduced polynomial is $P_{e_{i}}\left(\boldsymbol{z_{0}}\right)=\rho_{i}z_{j_{1}}^{m_{ij_{1}}}\ldots z_{j_{k_{i}}}^{m_{ij_{k_{i}}}}.$

To get it of degree $r_{\mathcal{H}}$ we add an additional variable
$y_{1}$ with multiplicity $m_{i\,n+1}=r_{\mathcal{H}}-\rho_{i}.$

The term $P_{e_{i}}\left(\boldsymbol{z_{0}}\right)$ with attached
tensor $\mathcal{P}_{e_{i}}$ of rank $\rho_{i}$ and dimension $n$
is transformed in: 
\[
R_{e_{i}}\left(\boldsymbol{z_{1}}\right)=P_{e_{i}}\left(\boldsymbol{z_{0}}\right)y_{1}^{m_{i\,n+1}}=r_{i}z_{j_{1}}^{m_{ij_{1}}}\ldots z_{j_{k_{i}}}^{m_{ij_{k_{i}}}}y_{1}^{m_{i\,n+1}}
\]
 with attached tensor $\mathcal{R}_{e_{i}}$ of rank $r_{\mathcal{H}}$
and dimension $n+1.$

The only non-zero elements of $\mathcal{R}_{e_{i}}$ are: 
\[
r_{j_{1}^{m_{ij_{1}}}\ldots j_{k_{i}}^{m_{ij_{k_{i}}}}\left(n+1\right)^{m_{i\,n+1}}}=\rho_{i}\dfrac{m_{ij_{1}}!\ldots m_{ij_{k_{i}}}!m_{i\,n+1}!}{r_{\mathcal{H}}!}
\]
 and all the elements of $\mathcal{R}_{e_{i}}$ obtained by permutation
of the indices and with same value. The number of permutation is:
\[
\dfrac{r_{\mathcal{H}}!}{m_{ij_{1}}!\ldots m_{ij_{k_{i}}}!m_{i\,n+1}!}.
\]

The hb-graph polynomial $P\left(\boldsymbol{z_{0}}\right)=\sum\limits _{i\in\left\llbracket p\right\rrbracket }c_{i}P_{e_{i}}\left(\boldsymbol{z_{0}}\right)$
is transformed into a homogeneous polynomial: 
\[
R\left(\boldsymbol{z_{1}}\right)=\sum\limits _{i\in\left\llbracket p\right\rrbracket }c_{i}R_{e_{i}}\left(\boldsymbol{z_{1}}\right)=\sum\limits _{i\in\left\llbracket p\right\rrbracket }c_{i}z_{j_{1}}^{m_{ij_{1}}}\ldots z_{j_{k_{i}}}^{m_{ij_{k_{i}}}}y_{1}^{m_{i\,n+1}}
\]
 representing the homogenized hb-graph $\mathcal{\overline{H}}$ with
attached tensor $\mathcal{R}=\sum\limits _{i=1}^{p}c_{e_{i}}\mathcal{R}_{e_{i}}$
where $c_{e_{i}}=\dfrac{r_{\mathcal{H}}}{\#_{m}e_{i}}.$

\begin{defin}

The \textbf{straightforward }$e$-adj\textbf{acency hypermatrix of
a hb-graph} $\mathcal{H}=\left(V,E\right)$ is the hypermatrix $\boldsymbol{A_{\text{str},\mathcal{H}}}$
defined by:

\[
\boldsymbol{A_{\text{str},\mathcal{H}}}=\sum\limits _{i\in\left\llbracket p\right\rrbracket }c_{e_{i}}\boldsymbol{R_{e_{i}}}.
\]

where for $e_{i}=\left\{ v_{j_{1}}^{m_{ij_{1}}},\ldots,v_{j_{k_{i}}}^{m_{ij_{k_{i}}}}\right\} \in E$
the associated hypermatrix is: $\boldsymbol{R_{e_{i}}}=\left(r_{i_{1}\ldots i_{r_{\mathcal{H}}}}\right)$
, which only non-zero elements are: 
\[
r_{j_{1}^{m_{ij_{1}}}\ldots j_{k_{i}}^{m_{ij_{k_{i}}}}\left(n+1\right)^{m_{i\,n+1}}}=\dfrac{m_{ij_{1}}!\ldots m_{ij_{k_{i}}}!m_{i\,n+1}!}{r_{\mathcal{H}}!}
\]
- with $m_{i\,n+1}=r_{\mathcal{H}}-\#_{m}e_{i}$ - and the ones with
same value and obtained by permutation of the indices and where $c_{e_{i}}=\dfrac{r_{\mathcal{H}}}{\#_{m}e_{i}}$.

\end{defin}

\subsubsection{Silo approach}

\paragraph*{Silo m-uniformisation:}

\begin{figure}
\begin{center}
\begin{tikzpicture}[->,>=stealth',scale=0.95, every node/.append style={transform shape}]
\node[state=green, minimum width=1cm, minimum height=1cm] (H) {$\mathcal{H}$};
\node[
		minimum width=0.1cm, 
		minimum height=1cm,
		right of=H,
		node distance=1.5cm] (bigBrackLeft) {$\Bigg(\setstackgap{S}{2pt}\raisebox{-1.2em}{\Shortstack{. . . . . . . . . .}}$};
\node[state=green, 
		minimum width=1cm, 
		minimum height=1cm,
		right of=bigBrackLeft,
		node distance=0.7cm] (Hr) {$\mathcal{H}_r$};
\node[right of=H,
		yshift=0.25cm,
		node distance=1cm] (phid) {$\phi_d$};
\node[state=green, 
		minimum width=1cm, 
		minimum height=1cm,
		right of=Hr,
		node distance=2cm] (Hr1) {$\mathcal{H}_{r,1}$};
\node[right of=Hr,
		yshift=0.25cm,
		node distance=1cm] (phicw) {$\phi_\textrm{cw}$};
\node[right of=Hr1,
		yshift=0.25cm,
		node distance=1cm] (phicd) {$\phi_\textrm{c-d}$};
\node[state=green, 
		minimum width=1cm, 
		minimum height=1cm,
		right of=Hr1,
		node distance=2cm] (Hrcr) {$\mathcal{H}_{r,c_r}$};
\node[state=orange!50!blue, 
		minimum width=5.2cm, 
		minimum height=2cm,
		right of=Hrcr,
		node distance=3.25cm] (spec) {};
\node[right of=Hrcr,
		yshift=0.75cm,
		node distance=1.75cm] (spec_text) {specific};
\node[right of=Hrcr,
		yshift=0.25cm,
		node distance=1.4cm] (phiyc) {$\phi_{y_{r}^{r_{\mathcal{H}}-r}\textrm{-v}}$};
\node[state=green, 
		minimum width=1cm, 
		minimum height=1cm,
		right of=Hrcr,
		node distance=2.7cm] (Htwd) {$\mathcal{H}^+_{r,c_r}$};
\node[
		minimum width=0.1cm, 
		minimum height=1cm,
		right of=Htwd,
		node distance=1.1cm] (bigBrackRight) {$\setstackgap{S}{2pt}\raisebox{-1.2em}{\Shortstack{. . . . . . . . . .}}\Bigg)_{r \in r_\mathcal{H}}$};
\node[right of=bigBrackRight,
		yshift=0.25cm,
		node distance=0.2cm] (phim) {$\phi_\textrm{m}$};
\node[state=green, 
		minimum width=1cm, 
		minimum height=1cm,
		right of=bigBrackRight,
		node distance=1.4cm] (Hwd) {$\widehat{\mathcal{H}_{\widehat{w}}}$};

\path (H) edge ([xshift=-0.1cm]bigBrackLeft.center)
(Hr) edge (Hr1)
(Hr1) edge (Hrcr)
(Hrcr) edge (Htwd)
([xshift=-0.25cm]bigBrackRight.center) edge (Hwd);

\end{tikzpicture}
\end{center}

\caption{Operations on the original hb-graph to m-uniformize it in the silo
approach. Parenthesis with vertical dots indicate parallel operations.}
\label{Fig: Operations uniformisation silo}
\end{figure}
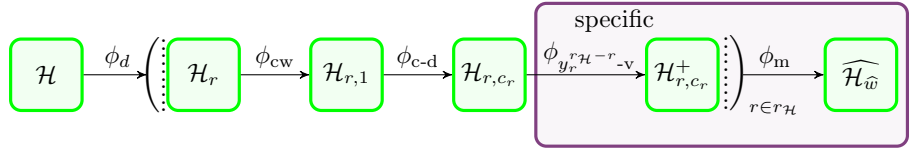

The first steps are similar to the straightforward approach.

The hb-graph $\mathcal{H}$ is decomposed in layers $\mathcal{H}=\bigoplus\limits _{r\in\left\llbracket r_{\mathcal{H}}\right\rrbracket }\mathcal{H}_{r}$
as described in sub-section \ref{subsec:Choice_technical}. Each $\mathcal{H}_{r},r\in\left\llbracket r_{\mathcal{H}}\right\rrbracket $
is canonically weighted and $c_{r}$-dilatated to obtain $\mathcal{H}_{r,c_{r}}$. 

We generate $r_{\mathcal{H}}-1$ new vertices $y_{i}\notin V$, $i\in\left\llbracket r_{\mathcal{H}}-1\right\rrbracket $.

We then apply to each $\mathcal{H}_{r,c_{r}},r\in\left\llbracket r_{\mathcal{H}}-1\right\rrbracket $
the $y_{r}^{r_{\mathcal{H}}-r}$-vertex-increasing operation to obtain
$\mathcal{H}_{r,c_{r}}^{+}$ the $y_{r}^{r_{\mathcal{H}}-r}$-complemented
hb-graph of each $\mathcal{H}_{r,c_{r}},r\in\left\llbracket r_{\mathcal{H}}-1\right\rrbracket $
. The family $\left(\mathcal{H}_{r,c_{r}}^{+}\right)_{r\in r_{\mathcal{H}}}$
is then merged using the merging operation to obtain the $r_{\mathcal{H}}$-m-uniform
hb-graph $\widehat{\mathcal{H}_{\widehat{w}}}$. 

The different steps are summarized in Figure \ref{Fig: Operations uniformisation silo}.

\begin{claim}The transformation $\phi_{s}:\mathcal{H}\mapsto\widehat{\mathcal{H}_{\widehat{w}}}$
preserves the $e$-adjacency.

\end{claim}

\begin{proof}
$\phi_{\text{s}}=\phi_{\text{m}}\circ
\left(\setstackgap{S}{2pt}\raisebox{-1.2em}{\Shortstack{. . . . . . . . . .}}\phi_{y_{r}^{r_{\mathcal{H}}-r}\text{-v}}\circ\phi_{\text{c-d}}\circ\phi_{\text{cw}}
\setstackgap{S}{2pt}\raisebox{-1.2em}{\Shortstack{. . . . . . . . . .}}\right)\circ\phi_{\text{d}}.$

The operations involved either preserve $e$-adjacency or preserve
exactly $e$-adjacency, also by composition $\phi_{\text{s}}$ preserve
$e$-adjacency.

\end{proof}

\paragraph*{Silo homogenization:}

In this homogenization process we suppose that the hb-edges are sorted
by m-cardinality.

We add $r_{\mathcal{H}}-1$ vertices $N_{1}$ to $N_{r_{\mathcal{H}}-1}$
into the universe, ie the vertex set, corresponding to $r_{\mathcal{H}}-1$
additional variables respectively $y_{1}$ to $y_{r_{\mathcal{H}}-1}$.

The term $P_{e_{i}}\left(\boldsymbol{z_{0}}\right)=z_{j_{1}}^{m_{ij_{1}}}\ldots z_{j_{k_{i}}}^{m_{ij_{k_{i}}}}$
of $P$ has degree the m-cardinality of the hb-edge $e_{i}$, ie $\#_{m}e_{i}$.
To get it of degree $r_{\mathcal{H}}$, we use the additional variable
$y_{\#_{m}e_{i}}$ with multiplicity $m_{i\,\#_{m}e_{i}}=r_{\mathcal{H}}-\#_{m}e_{i}$.

The term $P_{e_{i}}\left(\boldsymbol{z_{0}}\right)$ with attached
tensor $\mathcal{P}_{e_{i}}$ of rank $\#_{m}e_{i}$ and dimension
$n$ is transformed in $R_{e_{i}}\left(\boldsymbol{z_{\#_{m}e_{i}}}\right)=P_{e_{i}}\left(\boldsymbol{z_{0}}\right)y_{\#_{m}e_{i}}^{m_{i\,n+\#_{m}e_{i}}}$
with attached tensor $\mathcal{R}_{e_{i}}$ of rank $r_{\mathcal{H}}$
and dimension $n+1$

The only non-zero elements of $\mathcal{R}_{e_{i}}$ are: 
\[
r_{j_{1}^{m_{ij_{1}}}\ldots j_{k_{i}}^{m_{ij_{k_{i}}}}\left(n+\#_{m}e_{i}\right)^{m_{i\,n+\#_{m}e_{i}}}}=\dfrac{m_{ij_{1}}!\ldots m_{ij_{k_{i}}}!m_{i\,n+\#_{m}e_{i}}!}{r_{\mathcal{H}}!}
\]
 and the same value elements which indices are obtained by permutation
of this first element.

$P$ is transformed into a homogeneous polynomial 
\[
R\left(\boldsymbol{z_{r_{\mathcal{H}}-1}}\right)=\sum\limits _{i\in\left\llbracket p\right\rrbracket }c_{i}R_{e_{i}}\left(\boldsymbol{z_{\#_{m}e_{i}}}\right)=\sum\limits _{i\in\left\llbracket p\right\rrbracket }c_{i}z_{j_{1}}^{m_{ij_{1}}}\ldots z_{j_{k_{i}}}^{m_{ij_{k_{i}}}}y_{\#_{m}e_{i}}^{m_{i\,n+\#_{m}e_{i}}}
\]
 representing the homogenized hb-graph $\mathcal{\overline{H}}$ with
attached tensor $\mathcal{R}=\sum\limits _{i\in\left\llbracket p\right\rrbracket }c_{e_{i}}\mathcal{R}_{e_{i}}$
where: $c_{e_{i}}=\dfrac{r_{\mathcal{H}}}{\#_{m}e_{i}}$.

\begin{defin}

The silo $e$-adjacency hypermatrix of a hb-graph $\mathcal{H}=\left(V,E\right)$
is the hypermatrix $\boldsymbol{A_{\text{sil},\mathcal{H}}}=\left(a_{i_{1}\ldots i_{r_{\mathcal{H}}}}\right)_{i_{1},\ldots,i_{r_{\mathcal{H}}}\in\left\llbracket n\right\rrbracket }$
defined by:

\[
\boldsymbol{A_{\text{sil},\mathcal{H}}}=\sum\limits _{i\in\left\llbracket p\right\rrbracket }c_{e_{i}}\boldsymbol{R_{e_{i}}}
\]

and where for $e_{i}=\left\{ v_{j_{1}}^{m_{ij_{1}}},\ldots,v_{j_{k_{i}}}^{m_{ij_{k_{i}}}}\right\} \in E$
the associated hypermatrix is: $\boldsymbol{R_{e_{i}}}=\left(r_{i_{1}\ldots i_{r_{\mathcal{H}}}}\right)$
, which only non-zero elements are: 
\[
r_{j_{1}^{m_{ij_{1}}}\ldots j_{k_{i}}^{m_{ij_{k_{i}}}}\left(n+\#_{m}e_{i}\right)^{m_{i\,n+\#_{m}e_{i}}}}=\dfrac{m_{ij_{1}}!\ldots m_{ij_{k_{i}}}!m_{i\,n+\#_{m}e_{i}}!}{r_{\mathcal{H}}!}
\]
 and all elements of $\boldsymbol{R_{e_{i}}}$ obtained by permuting
\[
j_{1}^{m_{ij_{1}}}\ldots j_{k_{i}}^{m_{ij_{k_{i}}}}\left(n+\#_{m}e_{i}\right)^{m_{i\,n+\#_{m}e_{i}}},
\]
with: 
\[
m_{i\,n+\#_{m}e_{i}}=r_{\mathcal{H}}-\sum\limits _{l\in\left\llbracket k_{i}\right\rrbracket }m_{i\,j_{l}},
\]

and where: 
\[
c_{e_{i}}=\dfrac{r_{\mathcal{H}}}{\#_{m}e_{i}}.
\]

\end{defin}

\begin{rmk}

In this case, 
\[
\boldsymbol{A_{\text{sil},\mathcal{H}}}=\sum\limits _{r\in\left\llbracket r_{\mathcal{H}}\right\rrbracket }c_{r}\sum\limits _{e_{i}\in\left\{ e:\#_{m}e=r\right\} }\boldsymbol{R_{e_{i}}}
\]

where $c_{r}=\dfrac{r_{\mathcal{H}}}{r}.$

\end{rmk}

\subsubsection{Layered approach}

\paragraph*{Layered uniformisation:}

\begin{figure}
\begin{center}
\begin{tikzpicture}[->,>=stealth',scale=0.95, every node/.append style={transform shape}]
\node[state=green, minimum width=1cm, minimum height=1cm] (H) {$\mathcal{H}$};
\node[
		minimum width=0.1cm, 
		minimum height=1cm,
		right of=H,
		node distance=1.5cm] (bigBrackLeft) {$\Bigg(\setstackgap{S}{2pt}\raisebox{-1.2em}{\Shortstack{. . . . . . . . . .}}$};
\node[state=green, 
		minimum width=1cm, 
		minimum height=1cm,
		right of=bigBrackLeft,
		node distance=0.7cm] (Hr) {$\mathcal{H}_r$};
\node[right of=H,
		yshift=0.25cm,
		node distance=1cm] (phid) {$\phi_d$};
\node[state=green, 
		minimum width=1cm, 
		minimum height=1cm,
		right of=Hr,
		node distance=2cm] (Hr1) {$\mathcal{H}_{r,1}$};
\node[right of=Hr,
		yshift=0.25cm,
		node distance=1cm] (phicw) {$\phi_\textrm{cw}$};
\node[right of=Hr1,
		yshift=0.25cm,
		node distance=1cm] (phicd) {$\phi_\textrm{c-d}$};
\node[state=green, 
		minimum width=1cm, 
		minimum height=1cm,
		right of=Hr1,
		node distance=2cm] (Hrcr) {$\mathcal{H}_{r,c_r}$};
\node[state=orange!50!blue, 
		minimum width=12.5cm, 
		minimum height=5.5cm,
		right of=H,
		xshift=-0.25cm,
		yshift=-4.5cm,
		node distance=3.25cm] (spec) {};
\node[right of=spec,
		yshift=-2.5cm,
		node distance=5cm] (spec_text) {specific};
\node[
		minimum width=0.1cm, 
		minimum height=1cm,
		right of=Hrcr,
		node distance=1.1cm] (bigBrackRight) {$\setstackgap{S}{2pt}\raisebox{-1.2em}{\Shortstack{. . . . . . . . . .}}\Bigg)_{r \in r_\mathcal{H}}$};


\node[state=orange!50!white, 
		minimum width=4.8cm, 
		minimum height=3cm,
		right of=H,
		yshift=-4.25cm,
		xshift=1.8cm] (spec_iter) {};
\node[right of=spec_iter,
		yshift=-1.25cm,
		node distance=1cm] (spec_iter_text) {Iterative phase};
\node[state=orange!50!white, 
		minimum width=4.8cm, 
		minimum height=1.25cm,
		right of=H,
		yshift=-6.5cm,
		xshift=1.8cm] (spec_init) {};
\node[right of=spec_init,
		yshift=-0.4cm,
		xshift=-0.15cm,
		node distance=1cm] (spec_init_text) {Initialisation};
\node[state=green,
	right of=H,
	yshift=-3.5cm,
	minimum width=1cm, 
	minimum height=1cm] (Kk) {$\mathcal{K}_k$};

\node[state=green, 
		minimum width=1cm, 
		minimum height=1cm,
		right of=Kk,
		node distance=2.5cm] (Kk_plus) {$\mathcal{K}_k^+$};
\node[right of=Kk,
		yshift=0.25cm,
		node distance=1.25cm] (phiv) {$\phi_{y_k^1-v}$};
\node[state=green, 
		minimum width=1cm, 
		minimum height=1cm,
		right of=Kk,
		node distance=0.3cm,
		yshift=-1.5cm] (Hk_plus_un) {$\mathcal{H}_{k+1,c_{k+1}}$};

\node[diamond,
		fill=green!10,
		draw=green, very thick,
		right of=Kk_plus,
		node distance = 3cm,
		yshift = -1.5cm,
		minimum width=1cm, 
		minimum height=1cm] (condition) {$k<r_\mathcal{H}?$};		
\node[right of=Kk_plus,
		yshift=-1.25cm,
		node distance=1.35cm] (phim) {$\phi_{m}$};
\node[right of=condition,
		yshift=1cm,
		node distance=0.35cm] (yes) {$\textrm{yes}$};
\node[right of=condition,
		yshift=-0.25cm,
		node distance=1.25cm] (no) {$\textrm{no}$};

\node[right of=Kk,
		yshift=1.25cm,
		node distance=2.5cm] (iteration) {$k:=k+1$};
\node[state=green,
		right of=condition,
		node distance = 2cm,
		minimum width=1cm, 
		minimum height=1cm] (layered) {$\widehat{\mathcal{H}}$};
\node[diamond,
		fill=green!10,
		draw=green, very thick,
		right of=Kk,
		node distance = -2cm,
		minimum width=1cm, 
		minimum height=1cm] (cond_init) {$k>0?$};
\node[state=green,
		right of=cond_init,
		node distance = 2.4cm,
		yshift = -3cm,
		minimum width=1cm, 
		minimum height=1cm] (cond_init_k0) {$\mathcal{K}_0=\mathcal{H}_{1,c_1}$};
\node[right of=cond_init,
		node distance = -2.5cm] (start) {};
\node[right of=start,
		yshift=0.25cm,
		node distance=1.1cm] (yes_init) {$k:=0$};

\node[right of=cond_init,
		yshift=-1.75cm,
		node distance=0.25cm] (yes_init) {no};
\node[right of=cond_init,
		yshift=0.25cm,
		node distance=1cm] (no_init) {yes};

\path (H) edge ([xshift=-0.1cm]bigBrackLeft.center)
(Hr) edge (Hr1)
(Hr1) edge (Hrcr);
\draw[->]([xshift=-0.25cm]bigBrackRight.center) --++ (1,0) --++ (0,-1.25) --++ (-11,0) --++ (0,-2.25) --++ (cond_init);

\draw[->] (Kk) --++ (Kk_plus);
\draw[->,dashed] (cond_init) --++ (0,-3) --++ (cond_init_k0);
\draw[->,dashed] (cond_init) --++ (Kk);
\draw[->] (cond_init_k0) --++ (5.1,0) --++ (condition);
\draw[->,dashed] (condition) --++ (layered);
\draw[->] (Kk_plus) --++ (1,0) --++(0,-1.5) --++ (condition);
\draw[->] (Hk_plus_un) --++ (condition);
\draw[->,dashed] (condition) --++ (0,3) --++ (-7.5,0) --++ (cond_init);

\end{tikzpicture}
\end{center}

\caption{Operations on the original hb-graph to m-uniformize it in the layered
approach. Parenthesis with vertical dots indicate parallel operations.}
\label{Fig: Operations uniformisation silo-1}
\end{figure}
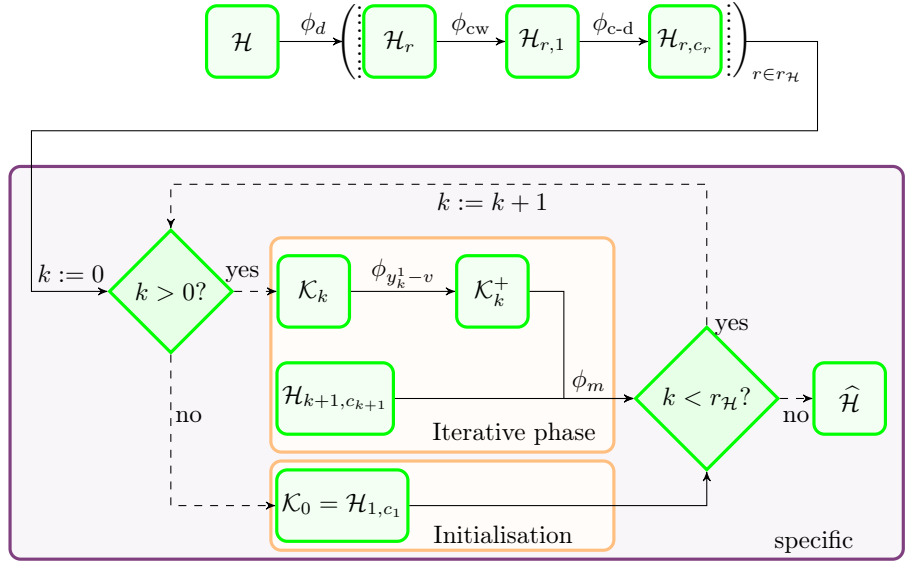

The first steps are similar to the straightforward approach.

The hb-graph $\mathcal{H}$ is decomposed in layers $\mathcal{H}=\bigoplus\limits _{r\in\left\llbracket r_{\mathcal{H}}\right\rrbracket }\mathcal{H}_{r}$
as described in sub-section \ref{subsec:Choice_technical}. Each $\mathcal{H}_{r},r\in\left\llbracket r_{\mathcal{H}}\right\rrbracket $
is canonically weighted and $c_{r}$-dilatated to obtain $\mathcal{H}_{r,c_{r}}$. 

We generate $r_{\mathcal{H}}-1$ new vertices $y_{i}\notin V$, $i\in\left\llbracket r_{\mathcal{H}}-1\right\rrbracket $
and write $V_{s}=\left\{ y_{i}:i\in\left\llbracket r_{\mathcal{H}}-1\right\rrbracket \right\} $

A two-phase steps iteration as it has been done with hypergraphs in
\citet{ouvrard2017cooccurrence} and \citet{ouvrard2018imaadjacency}
is considered: the inflation phase (IP) and the merging phase (MP).
At step $k=0$, $\mathcal{K}_{0}=\mathcal{H}_{1,c_{1}}$ and no further
action is made but increasing $k$ of 1 and going to the next step.
At step $k>0$, the input is the $k$-m-uniform weighted hb-graph
$\mathcal{K}_{k}$obtained from the previous iteration. In the IP,
$\mathcal{K}_{k}$ is transformed into $\mathcal{K}_{k}^{+}$ the
$y_{k}^{1}$-vertex-increased hb-graph, which is $\left(k+1\right)$-m-uniform. 

The MP is merging the hypergraphs $\mathcal{K}_{k}^{+}$ and $\mathcal{H}_{k+1,c_{k+1}}$
into a single $\left(k+1\right)$-m-uniform hb-graph $\widehat{\mathcal{K}_{\widehat{w}}}.$ 

We iterate while $k<r_{\mathcal{H}}$, increasing in between each
step $k$ of 1. When $k$ reaches $r_{\mathcal{H}}$, we stop iterating
and the last $\widehat{\mathcal{K}_{\widehat{w}}}$ obtained, written
$\widehat{\mathcal{H}_{\widehat{w}}}$ is called the $V_{S}$-layered
m-uniform hb-graph of $\mathcal{H}.$

The different steps are summarized in Figure \ref{Fig: Operations uniformisation silo-1}.

\begin{claim}The transformation $\phi_{s}:\mathcal{H}\mapsto\widehat{\mathcal{H}_{\widehat{w}}}$
preserves the $e$-adjacency.

\end{claim}

\begin{proof}
$\phi_{\text{s}}=\psi\circ
\left(\setstackgap{S}{2pt}\raisebox{-1.2em}{\Shortstack{. . . . . . . . . .}}\phi_{\text{c-d}}\circ\phi_{\text{cw}}
\setstackgap{S}{2pt}\raisebox{-1.2em}{\Shortstack{. . . . . . . . . .}}\right)\circ\phi_{\text{d}}$, where $\psi$ is called the iterative layered operation that converts
the family obtained by 
$\left(\setstackgap{S}{2pt}\raisebox{-1.2em}{\Shortstack{. . . . . . . . . .}}\phi_{\text{c-d}}\circ\phi_{\text{cw}}
\setstackgap{S}{2pt}\raisebox{-1.2em}{\Shortstack{. . . . . . . . . .}}\right)\circ\phi_{\text{d}}$ and transform it into the $V_{S}$-layered m-uniform hb-graph of
$\mathcal{H}$.

The operations involved in the operations $\phi_{\text{c-d}}$, $\phi_{\text{c-w}}$
and $\phi_{\textrm{d}}$ either preserve $e$-adjacency or preserve
exactly $e$-adjacency, and so forth by composition.

The iterative layered operation preserve e-adjacency as the operations
involved are preserving e-adjacency and that the family of hb-graphs
at the input has hb-edges family that are totally distinct. 

Also by composition $\phi_{\text{s}}$ preserve $e$-adjacency.

\end{proof}

\paragraph{Layered homogenization:}

This solution was first developed in \citet{ouvrard2017cooccurrence}
for general hypergraphs. The idea is to sort the hb-edges as in the
silo homogenization and considering as well $r_{\mathcal{H}}-1$ additional
vertices $L_{1}$ to $L_{r_{\mathcal{H}}-1}$ into the universe, corresponding
to $r_{\mathcal{H}}-1$ additional variables respectively $y_{1}$
to $y_{r_{\mathcal{H}}-1}$.

But these vertices are added successively to each hb-edge to fill
the hb-edges to a $r_{\mathcal{H}}$ value of the $m$-cardinality:
a hb-edge of initial cardinality $\#_{m}e$ will be filled with elements
$L_{\#_{m}e}$ to $L_{r_{\mathcal{H}}-1}$. It matches to add the
$k$-m-uniform hb-subgraph $\mathcal{H}_{k}$ with the $k+1$-m-uniform
hb-subgraph $\mathcal{H}_{k+1}$ by filling the hb-edge of $\mathcal{H}_{k}$
with the additional vertex $L_{k}$ to get a homogenised $k+1$-m-uniform
hb-subgraph of the homogenised hb-graph $\overline{\mathcal{H}}$.

A hb-edge of $m$-cardinality $\#_{m}e_{i}$ is represented by the
polynomial 
\[
P_{e_{i}}\left(\boldsymbol{z_{0}}\right)=z_{j_{1}}^{m_{ij_{1}}}\ldots z_{j_{k_{i}}}^{m_{ij_{k_{i}}}}
\]
 of degree $\#_{m}e_{i}$. 

All the hb-edges of same m-cardinality $m$ belongs to the same layer
of level $m$. To transform the hb-edge of m-cardinality $\#_{m}e_{i}+1$
we fill it with the element $L_{\#_{m}e_{i}}$. 

In this case, the polynomial $P_{e_{i}}\left(\boldsymbol{z_{0}}\right)$
is transformed into:
\[
R_{(1)e_{i}}\left(\boldsymbol{z_{\#_{m}e_{i}}}\right)=P_{e_{i}}\left(\boldsymbol{z_{0}}\right)y_{\#_{m}e_{i}}^{1}
\]
 of degree $\#_{m}e_{i}+1$. 

Iterating over the layers the polynomial: 
\[
P_{e_{i}}\left(\boldsymbol{z_{0}}\right)=z_{j_{1}}^{m_{ij_{1}}}\ldots z_{j_{k_{i}}}^{m_{ij_{k_{i}}}}
\]
 is transformed in: 
\[
R_{(r_{\mathcal{H}}-\#_{m}e_{i})e_{i}}\left(\boldsymbol{z_{r_{\mathcal{H}}-1}}\right)=P_{e_{i}}\left(\boldsymbol{z_{0}}\right)y_{\#_{m}e_{i}}^{1}\ldots y_{r_{\mathcal{H}}-1}^{1}
\]
 of degree $r_{\mathcal{H}}$.

The polynomial $P_{e_{i}}\left(\boldsymbol{z_{0}}\right)$ with attached
tensor $\mathcal{P}_{e_{i}}$ of rank $\#_{m}e_{i}$ and dimension
$n$ is transformed in: 
\[
R_{(r_{\mathcal{H}}-\#_{m}e_{i})e_{i}}\left(\boldsymbol{z_{r_{\mathcal{H}}-1}}\right)=R_{e_{i}}\left(\boldsymbol{z_{0}}\right)y_{\#_{m}e_{i}}^{1}\ldots y_{r_{\mathcal{H}}-1}^{1}
\]
 with attached tensor $\mathcal{R}_{\left(r_{\mathcal{H}}-\#_{m}e_{i}\right)e_{i}}$
of rank $r_{\mathcal{H}}$ and dimension $n+r_{\mathcal{H}}-1$.

The only non-zero elements of $\mathcal{R}_{(r_{\mathcal{H}}-\#_{m}e_{i})e_{i}}$
are: 
\[
r_{(r_{\mathcal{H}}-\#_{m}e_{i})\,j_{1}^{m_{ij_{1}}}\ldots j_{k_{i}}^{m_{ij_{k_{i}}}}\left[n+\#_{m}e_{i}\right]^{1}\ldots\left[n+r_{\mathcal{H}}-1\right]^{1}}=\dfrac{m_{ij_{1}}!\ldots m_{ij_{k_{i}}}!}{r_{\mathcal{H}}!}
\]
 and all the elements of $\mathcal{R}_{(r-\#_{m}e_{i})e_{i}}$ obtained
by permuting:
\[
j_{1}^{m_{ij_{1}}}\ldots j_{k_{i}}^{m_{ij_{k_{i}}}}\left[n+\#_{m}e_{i}\right]^{1}\ldots\left[n+r_{\mathcal{H}}-1\right]^{1}
\]

And $P$ is transformed in a homogeneous polynomial 
\begin{eqnarray*}
R\left(\boldsymbol{z_{r_{\mathcal{H}}-1}}\right) & = & \sum\limits _{i\in\left\llbracket p\right\rrbracket }c_{i}R_{(r_{\mathcal{H}}-\#_{m}e_{i})e_{i}}\left(\boldsymbol{z_{r_{\mathcal{H}}-1}}\right)\\
 & = & \sum\limits _{i\in\left\llbracket p\right\rrbracket }c_{i}z_{j_{1}}^{m_{ij_{1}}}\ldots z_{j_{k_{i}}}^{m_{ij_{k_{i}}}}y_{\#_{m}e_{i}}^{1}\ldots y_{r_{\mathcal{H}}-1}^{1}
\end{eqnarray*}
 representing the homogenized hb-graph $\mathcal{\overline{H}}$ with
attached tensor 
\[
\mathcal{R}=\sum\limits _{i\in\left\llbracket p\right\rrbracket }c_{e_{i}}\mathcal{R}_{(r_{\mathcal{H}}-\#_{m}e_{i})e_{i}},
\]

where: 
\[
c_{e_{i}}=\dfrac{r_{\mathcal{H}}}{\#_{m}e_{i}}.
\]

\begin{defin}

The layered $e$-adjacency tensor of a hb-graph $\mathcal{H}=\left(V,E\right)$
is the tensor 
\[
\mathcal{A}_{\text{lay}}\left(\mathcal{H}\right)=\left(a_{i_{1}\ldots i_{r_{\mathcal{H}}}}\right)_{1\leqslant i_{1},\ldots,i_{r_{\mathcal{H}}}\leqslant n}
\]
 defined by:

\[
\mathcal{A}_{\text{lay}}\left(\mathcal{H}\right)=\sum\limits _{i\in\left\llbracket p\right\rrbracket }c_{e_{i}}\mathcal{R}_{\left(r_{\mathcal{H}}-\#_{m}e_{i}\right)e_{i}}
\]

where for $e_{i}=\left\{ v_{j_{1}}^{m_{ij_{1}}},\ldots,v_{j_{k_{i}}}^{m_{ij_{k_{i}}}}\right\} \in E$
the associated tensor is: 
\[
\mathcal{R}_{\left(r_{\mathcal{H}}-\#_{m}e_{i}\right)e_{i}}=\left(r_{\left(r_{\mathcal{H}}-\#_{m}e_{i}\right)i_{1}\ldots i_{r_{\mathcal{H}}}}\right),
\]
which only non-zero elements are: 
\[
r_{\left(r_{\mathcal{H}}-\#_{m}e_{i}\right)j_{1}^{m_{ij_{1}}}\ldots j_{k_{i}}^{m_{ij_{k_{i}}}}\left[n+\#_{m}e_{i}\right]^{1}\ldots\left[n+r_{\mathcal{H}}-1\right]^{1}}=\dfrac{m_{ij_{1}}!\ldots m_{ij_{k_{i}}}!}{r_{\mathcal{H}}!}
\]
 and all elements of $\mathcal{R}_{e_{i}}$with same value obtained
by permuting:
\[
j_{1}^{m_{ij_{1}}}\ldots j_{k_{i}}^{m_{ij_{k_{i}}}}\left[n+\#_{m}e_{i}\right]^{1}\ldots\left[n+r_{\mathcal{H}}-1\right]^{1},
\]

and where: $c_{e_{i}}=\dfrac{r_{\mathcal{H}}}{\#_{m}e_{i}}$.

\end{defin}

\begin{rmk}

$\mathcal{A}_{\text{lay}}\left(\mathcal{H}\right)$ can also be written:
\[
\mathcal{A}_{\text{lay}}\left(\mathcal{H}\right)=\sum\limits _{r\in\left\llbracket r_{\mathcal{H}}\right\rrbracket }c_{r}\sum\limits _{e_{i}\in\left\{ e:\#_{m}e=r\right\} }\mathcal{R}_{e_{i}},
\]

where $c_{r}=\dfrac{r_{\mathcal{H}}}{r}.$

\end{rmk}

\section{Results on the constructed tensors}

\label{sec:Results_constructed_tensors}

Each of the tensor built is of rank $r_{\mathcal{H}}$ and of dimension
$n+n_{\mathcal{A}}$ where $n_{\mathcal{A}}$ is: 
\begin{itemize}
\item in the straightforward approach: $n_{\mathcal{A}}=1$.
\item in the silo approach and the layered approach: $n_{A}=r_{\mathcal{H}}-1$.
\end{itemize}

\subsection{Information on hb-graph}

\subsubsection{m-degree of vertices}

We built the different tensors in a way that the retrieval of the
vertex m-degree is possible; the null vertex(-ices) added give(s)
additional information on the structure of the hb-graph.

\begin{claim}

Let consider for $j\in\llbracket n\rrbracket$ a vertex $v_{j}\in V$.

Then in each of the $e$-adjacency tensors built, it holds: 
\[
\sum\limits _{j_{2},...,j_{r_{\mathcal{H}}}\in\left\llbracket n+n_{\mathcal{A}}\right\rrbracket }a_{jj_{2}\ldots j_{r_{\mathcal{H}}}}=\sum\limits _{i\,:\,v_{j}\in e_{i}}m_{i\,j}=\text{deg}_{m}\left(v_{j}\right)
\]

\end{claim}

\begin{proof}

For $j\in\left\llbracket n\right\rrbracket $:

$\sum\limits _{j_{2},...,j_{r_{\mathcal{H}}}\in\left\llbracket n+n_{\mathcal{A}}\right\rrbracket }a_{jj_{2}\ldots j_{r_{\mathcal{H}}}}$
has non-zero terms only for corresponding hb-edges of original hb-graph
$e_{i}$ that have $v_{j}$ in it. For such a hb-edge containing $v_{j}$,
it is described by $e_{i}=\left\{ v_{j}^{m_{i\,j}},v_{l_{2}}^{m_{i\,l_{2}}},\ldots,v_{l_{k}}^{m_{i\,l_{k}}}\right\} $.
It means that the multiset $\left\{ \left\{ j_{2},\ldots,j_{r_{\mathcal{H}}}\right\} \right\} $
corresponds exactly to the multiset $\left\{ j^{m_{i\,j}-1},l_{2}^{m_{i\,l_{2}}},\ldots,l_{k}^{m_{i\,l_{k}}}\right\} $. 

In the straightforward approach, for each $e_{i}$ such that $v_{j}\in e_{i},$
there are: 
\[
\dfrac{\left(r_{\mathcal{H}}-1\right)!}{\left(m_{i\,j}-1\right)!m_{i\,l_{2}}!\ldots m_{i\,l_{k}}!m_{i\,n+1}!}
\]
 possible permutations of the indices $j_{2}$ to $j_{r_{\mathcal{H}}}$
and 
\[
a_{jj_{2}...j_{r_{\mathcal{H}}}}=\dfrac{m_{i\,j}!m_{i\,l_{2}}!\ldots m_{i\,l_{k}}!m_{i\,n+1}!}{\left(r_{\mathcal{H}}-1\right)!}.
\]

In the silo approach, for each $e_{i}$ such that $v_{j}\in e_{i},$
there are 
\[
\dfrac{\left(r_{\mathcal{H}}-1\right)!}{\left(m_{i\,j}-1\right)!m_{i\,l_{2}}!\ldots m_{i\,l_{k}}!m_{i\,n+\#_{m}e_{i}}!}
\]
 possible permutations of the indices $j_{2}$ to $j_{r_{\mathcal{H}}}$
and 
\[
a_{jj_{2}...j_{r_{\mathcal{H}}}}=\dfrac{m_{i\,j}!m_{i\,l_{2}}!\ldots m_{i\,l_{k}}!m_{i\,n+\#_{m}e_{i}}!}{\left(r_{\mathcal{H}}-1\right)!}.
\]

In the layered approach, for each $e_{i}$ such that $v_{j}\in e_{i},$
there are: 
\[
\dfrac{\left(r_{\mathcal{H}}-1\right)!}{\left(m_{i\,j}-1\right)!m_{i\,l_{2}}!\ldots m_{i\,l_{k}}!}
\]
 possible permutations of the indices $j_{2}$ to $j_{r_{\mathcal{H}}}$
which have all the same value equals to: 
\[
a_{jj_{2}...j_{r_{\mathcal{H}}}}=\dfrac{m_{i\,j}!m_{i\,l_{2}}!\ldots m_{i\,l_{k}}!}{\left(r_{\mathcal{H}}-1\right)!}.
\]

Also, whatever the approach taken: 
\[
\sum\limits _{j_{2},...,j_{r_{\mathcal{H}}}\in\left\llbracket n\right\rrbracket }a_{jj_{2}\ldots j_{r_{\mathcal{H}}}}=\sum\limits _{i\,:\,v_{j}\in e_{i}}m_{i\,j}=\text{deg}_{m}\left(v_{j}\right).
\]

\end{proof}

\subsubsection{Additional vertex information}

The additional vertices carry information on the hb-edges of the hb-graph:
the information carried depends on the approach taken.

\begin{claim}

The layered $e$-adjacency tensor allows the retrieval of the distribution
of the hb-edges.

\end{claim}

\begin{proof}

For $j\in\left\llbracket n_{\mathcal{A}}\right\rrbracket $:

$\sum\limits _{j_{2},...,j_{r_{\mathcal{H}}}\in\left\llbracket n+n_{\mathcal{A}}\right\rrbracket }a_{n+jj_{2}\ldots j_{r_{\mathcal{H}}}}$
has non-zero terms only for corresponding hb-edges of the uniformized
hb-graph $\overline{e_{i}}$ that have $v_{j}$ in it. Such a hb-edge
is described by:

\[
\overline{e_{i}}=\left\{ v_{k}^{m_{i\,k}}\colon1\leqslant k\leqslant n+n_{\mathcal{A}}\right\} .
\]

It means that the multiset: 
\[
\left\{ \left\{ j_{2},\ldots,j_{r_{\mathcal{H}}}\right\} \right\} 
\]
 corresponds exactly to the multiset: 
\[
\left\{ \left(n+j\right)^{m_{i\,n+j}-1}\right\} +\left\{ k^{m_{i\,k}}\colon1\leqslant k\leqslant n+n_{\mathcal{A}},k\neq j\right\} .
\]

The number of possible permutations of elements in this multiset is:
\[
\dfrac{\left(r_{\mathcal{H}}-1\right)!}{\left(m_{i\,n+j}-1\right)!\prod\limits _{\substack{k\in\left\llbracket n\right\rrbracket }
}m_{i\,k}!\prod\limits _{\substack{k\in\left\llbracket n+1;n+n_{\mathcal{A}}\right\rrbracket \\
k\neq j
}
}m_{i\,k}!}
\]
 and the elements corresponding to one hb-edge are all equals to:
\[
\dfrac{\prod\limits _{\substack{k\in\left\llbracket n_{\mathcal{A}}\right\rrbracket }
}m_{i\,k}!}{\left(r_{\mathcal{H}}-1\right)!}.
\]

Thus: $\sum\limits _{j_{2},...,j_{r_{\mathcal{H}}}\in\left\llbracket n+n_{\mathcal{A}}\right\rrbracket }a_{n+jj_{2}\ldots j_{r_{\mathcal{H}}}}=\sum\limits _{j_{2},...,j_{r_{\mathcal{H}}}\in\left\llbracket n\right\rrbracket }m_{i\,n+j}=\deg_{m}\left(N_{j}\right)$

The interpretation differs between the different approaches.

\textbf{For the silo approach:}

There is one added vertex in each hb-edge. The silo of hb-edges of
m-cardinality $m_{s}$ $\left(m_{s}\in\left\llbracket r_{\mathcal{H}}-1\right\rrbracket \right)$
is associated to the null vertex $N_{m_{s}}$. The multiplicity of
$N_{m_{s}}$ in each hb-edge of the silo is $r_{\mathcal{H}}-m_{s}$.

Hence: 
\[
\dfrac{\deg_{m}\left(N_{j}\right)}{r_{\mathcal{H}}-m_{s}}=\left|\left\{ e:\#_{m}e=m_{s}\right\} \right|.
\]

The number of hb-edges in the silo $m_{s}$ is then deduced by the
following formula:

\[
\left|\left\{ e:\#_{m}e=m_{s}\right\} \right|=\left|E\right|-\sum\limits _{m_{s}\in\left\llbracket r_{\mathcal{H}}-1\right\rrbracket }\dfrac{\deg_{m}\left(N_{j}\right)}{r_{\mathcal{H}}-m_{s}}.
\]

\textbf{For the layered approach:}

The vertex $N_{j}$ corresponds to the layer of level $j$ and added
to each hb-edge that has m-cardinality less or equal to $j$ with
a multiplicity of 1.

Also: 
\[
\deg_{m}\left(N_{j}\right)=\left|\left\{ e:\#_{m}e\leqslant j\right\} \right|.
\]

Hence, for $j\in\left\llbracket 2;r_{\mathcal{H}}-1\right\rrbracket $:

\[
\left|\left\{ e:\#_{m}e=j\right\} \right|=\deg_{m}\left(N_{j}\right)-\deg_{m}\left(N_{j-1}\right)
\]

and:

\[
\left|\left\{ e:\#_{m}e=1\right\} \right|=\deg_{m}\left(N_{1}\right)
\]

\[
\left|\left\{ e:\#_{m}e=r_{\mathcal{H}}\right\} \right|=\left|E\right|-\deg_{m}\left(N_{r_{\mathcal{H}}-1}\right)
\]

\textbf{For the straightforward approach}:

In a hb-edge of m-cardinality $j\in\left\llbracket r_{\mathcal{H}}-1\right\rrbracket $,
the vertex $N_{1}$ is added in multiplicity $r_{\mathcal{H}}-j$.
The number of hb-edge of m-cardinality $j$ can be retrieved by considering
the elements of $\mathcal{A}_{\text{str}}\left(\mathcal{H}\right)$
of index $(n+1)i_{1}\ldots i_{r_{\mathcal{H}}-1}$ where $1\leqslant i_{1}\leqslant...\leqslant i_{j}\leqslant n$
and $i_{j+1}=\ldots=i_{r_{\mathcal{H}}-1}=n+1$ and thus of indices
obtained by permutation.

It follows for $j\in\left\llbracket r_{\mathcal{H}}-1\right\rrbracket $:

\begin{eqnarray*}
\left|\left\{ e:\#_{m}e=j\right\} \right| & = & \left|\left\{ e\colon N_{1}\in e\land m_{e}\left(N_{1}\right)=r_{\mathcal{H}}-j\right\} \right|\\
 & = & \sum_{\begin{array}{c}
i_{1},\ldots,i_{r_{\mathcal{H}}-1}\in\left\llbracket n+1\right\rrbracket \\
\left|\left\{ i_{k}=n+1\right\} \right|=r_{\mathcal{H}}-j-1
\end{array}}a_{n+1i_{1}\ldots i_{r_{\mathcal{H}}-1}}
\end{eqnarray*}

The terms of this sum $a_{n+1i_{1}\ldots i_{r_{\mathcal{H}}-1}}$
are non-zero only for corresponding hb-edges $\overline{e}$ of the
uniformized hb-graph that have $N_{1}$ in multiplicity $r_{\mathcal{H}}-j$
in it. Such a hb-edge is described by:

\[
\overline{e_{i}}=\left\{ v_{k}^{m_{i\,k}}\colon1\leqslant k\leqslant n\right\} +\left\{ N_{1}^{r_{\mathcal{H}}-j}\right\} .
\]

It means that the multiset: 
\[
\left\{ \left\{ i_{1},\ldots,i_{r_{\mathcal{H}}-1}\right\} \right\} 
\]
 corresponds exactly to the multiset: 
\[
\left\{ k^{m_{i\,k}}\colon k\in\left\llbracket n\right\rrbracket \right\} +\left\{ n+1^{r_{\mathcal{H}}-j-1}\right\} .
\]

The number of possible permutations of elements in this multiset is:
\[
\dfrac{\left(r_{\mathcal{H}}-1\right)!}{\prod\limits _{\substack{k\in\left\llbracket n\right\rrbracket }
}m_{i\,k}!\left(r_{\mathcal{H}}-j-1\right)!}
\]
 and the elements corresponding to one hb-edge all equal: 
\[
\dfrac{\prod\limits _{\substack{k\in\left\llbracket n\right\rrbracket }
}m_{i\,k}!\times\left(r_{\mathcal{H}}-j\right)!}{\left(r_{\mathcal{H}}-1\right)!}.
\]

Hence: 
\[
\dfrac{1}{r_{\mathcal{H}}-j}\sum\limits _{\begin{array}{c}
i_{2},\ldots,i_{r_{\mathcal{H}}}\in\left\llbracket n+1\right\rrbracket \\
\left|\left\{ i_{k}=n+1:2\leqslant k\leqslant r_{\mathcal{H}}\right\} \right|=r_{\mathcal{H}}-j-1
\end{array}}a_{n+1i_{2}\ldots i_{r_{\mathcal{H}}}}=\left|\left\{ e:\#_{m}e=j\right\} \right|
\]

The number of hb-edges of m-cardinality $r_{\mathcal{H}}$ can be
retrieved by:

\[
\left|\left\{ e:\#_{m}e=r_{\mathcal{H}}\right\} \right|=\left|E\right|-\sum\limits _{j\in\left\llbracket r_{\mathcal{H}}-1\right\rrbracket }\left|\left\{ e:\#_{m}e=j\right\} \right|.
\]

\end{proof}

\subsection{Initial results on spectral analysis}

Let $\mathcal{H}=\left(V,E\right)$ be a general hb-graph of $e$-adjacency
tensor $\mathcal{A}_{\mathcal{H}}=\left(a_{i_{1}...i_{k_{\max}}}\right)$
of order $k_{\max}$ and dimension $n+n_{\mathcal{A}}$

In the $e$-adjacency tensor $\mathcal{A}_{\mathcal{H}}$ built, the
diagonal entries are no longer equal to zero. As all elements of $\mathcal{A}_{\mathcal{H}}$
are all non-negative real numbers and as we have shown that: 
\[
\sum\limits _{\substack{i_{2},...,i_{m}\in\left\llbracket n+n_{\mathcal{A}}\right\rrbracket \\
\delta_{ii_{2}...i_{m}=0}
}
}a_{ii_{2}...i_{r_{\mathcal{H}}}}=\begin{cases}
d_{i} & \text{if}\,i\in\left\llbracket n\right\rrbracket \\
d_{n+j} & \text{if}\,i=n+j,\,j\in\left\llbracket n_{\mathcal{A}}\right\rrbracket .
\end{cases}
\]

It follows:

\begin{claim}The $e$-adjacency tensor $\mathcal{A}_{\mathcal{H}}=\left(a_{i_{1}...i_{k_{\max}}}\right)$
of a general hypergraph $\mathcal{H}=\left(V,E\right)$ has its eigenvalues
$\lambda$ such that: 
\begin{equation}
\left|\lambda\right|\leqslant\max\left(\Delta,\Delta^{\star}\right)+r_{\mathcal{H}}\label{eq:bound_max_degree_layer}
\end{equation}
 where $\Delta=\underset{i\in\left\llbracket n\right\rrbracket }{\max}\left(d_{i}\right)$
and $\Delta^{\star}=\underset{i\in\left\llbracket n_{\mathcal{A}}\right\rrbracket }{\max}\left(d_{n+i}\right)$

\end{claim}

\begin{proof}

From 
\begin{equation}
\forall i\in\left\llbracket 1,n\right\rrbracket ,\,\left(\mathcal{A}x^{m-1}\right)_{i}=\lambda x_{i}^{m-1}\label{eq:eigenvalue}
\end{equation}
we can write as $a_{ii_{2}...i_{m}}$are non-negative real numbers,
that for all $\lambda$ it holds: 
\begin{equation}
\left|\lambda-a_{i\ldots i}\right|\leqslant\sum\limits _{\substack{i_{2},...,i_{m}\in\left\llbracket n+n_{\mathcal{A}}\right\rrbracket \\
\delta_{ii_{2}...i_{m}=0}
}
}a_{ii_{2}...i_{m}}\label{eq:lambda_radius}
\end{equation}

Considering the triangular inequality: 
\begin{equation}
\left|\lambda\right|\leqslant\left|\lambda-a_{i\ldots i}\right|+\left|a_{i\ldots i}\right|\label{eq:lambda_triangular}
\end{equation}

Combining \ref{eq:lambda_radius} and \ref{eq:lambda_triangular}
yield:

\begin{equation}
\left|\lambda\right|\leqslant\sum\limits _{\substack{i_{2},...,i_{m}\in\left\llbracket n+n_{\mathcal{A}}\right\rrbracket \\
\delta_{ii_{2}...i_{m}=0}
}
}a_{ii_{2}...i_{m}}+\left|a_{i\ldots i}\right|.\label{eq:maj_lambda}
\end{equation}

But, whatever the approach taken, if $\left\{ i^{r_{\mathcal{H}}}\right\} $
is an hb-edge of the hb-graph

\[
\left|a_{i\ldots i}\right|=r_{\mathcal{H}}
\]

otherwise: 

\[
\left|a_{i\ldots i}\right|=0
\]

and thus writing $\Delta=\underset{i\in\left\llbracket n\right\rrbracket }{\max}\left(\deg_{m}\left(v_{i}\right)\right)$
and $\Delta^{\star}=\underset{i\in\left\llbracket n_{\mathcal{A}}\right\rrbracket }{\max}\left(\deg_{m}\left(N_{i}\right)\right)$
and using \ref{eq:maj_lambda} yield: 
\[
\left|\lambda\right|\leqslant\max\left(\Delta,\Delta^{\star}\right)+r_{\mathcal{H}}.
\]

\end{proof}

\begin{rmk}

In the straightforward approach: 
\begin{eqnarray*}
\Delta^{\star} & = & \deg_{m}\left(N_{1}\right)\\
 & = & \sum\limits _{j\in\left\llbracket r_{\mathcal{H}}-1\right\rrbracket }\left(r_{\mathcal{H}}-j\right)\left|\left\{ e:\#_{m}e=j\right\} \right|
\end{eqnarray*}

In the silo approach:

\begin{eqnarray*}
\Delta^{\star} & = & \underset{j\in\left\llbracket r_{\mathcal{H}}-1\right\rrbracket }{\max}\left(\deg_{m}\left(N_{j}\right)\right)\\
 & = & \underset{j\in\left\llbracket r_{\mathcal{H}}-1\right\rrbracket }{\max}\left(\left(r_{\mathcal{H}}-j\right)\left|\left\{ e:\#_{m}e=j\right\} \right|\right)
\end{eqnarray*}

In the layered approach:

\begin{eqnarray*}
\Delta^{\star} & = & \underset{j\in\left\llbracket r_{\mathcal{H}}-1\right\rrbracket }{\max}\left(\deg_{m}\left(N_{j}\right)\right)\\
 & = & \underset{j\in\left\llbracket r_{\mathcal{H}}-1\right\rrbracket }{\max}\left(\left|\left\{ e:\#_{m}e\leqslant j\right\} \right|\right)\\
 & = & \left|\left\{ e:\#_{m}e\leqslant r_{\mathcal{H}}-1\right\} \right|
\end{eqnarray*}

The values of $\Delta$ don't change whatever the approach taken is.

\end{rmk}

\section{Evaluation and final choice}

\label{sec:Evaluation}

\subsection{Evaluation}

We have put together some key features of the $e$-adjacency tensors
proposed in this article: the one of the straightforward approach
$\mathcal{A_{\text{str}}\left(H\right)},$the one of the silo approach
$\mathcal{A_{\text{sil}}\left(H\right)}$ and the one of the layered
approach $\mathcal{A_{\text{lay}}\left(H\right)}$. 

The constructed tensors have all same order $r_{\mathcal{H}}$. $\mathcal{A_{\text{sil}}\left(H\right)}$
and $\mathcal{A_{\text{lay}}\left(H\right)}$ dimensions are $r_{\mathcal{H}}-2$
bigger than $\mathcal{A_{\text{str}}\left(H\right)}$ ($n-2$ in the
worst case). $\mathcal{A_{\text{str}}\left(H\right)}$ has a total
number of elements $\dfrac{\left(n+1\right)^{r_{\mathcal{H}}}}{\left(n+r_{\mathcal{H}}-1\right)^{r_{\mathcal{H}}}}$
times smaller than the two other tensors. 

Elements of $\mathcal{A_{\text{str}}\left(H\right)}$ - respectively
$\mathcal{A_{\text{sil}}\left(H\right)}$ - are repeated $\dfrac{1}{n_{j}!}$
- respectively $\dfrac{1}{n_{j\,k}!}$ - times less than elements
of $\mathcal{A_{\text{lay}}\left(H\right)}$. The total number of
non nul elements filled for a given hb-graph in $\mathcal{A_{\text{str}}\left(H\right)}$
and $\mathcal{A_{\text{sil}}\left(H\right)}$ are the same and is
smaller than the total number of non nul elements in $\mathcal{A_{\text{lay}}\left(H\right)}$

The number of elements to be filled before permutation to have full
description of a hb-edge is constant and equals to 1 whatever the
approach taken and the value depends only on the hb-edge composition.

All tensors are symmetric and allow reconstructivity of the hb-graph
from the elements.

Nodes degree can be retrieved as it has been shown previously, but
it is easier with the silo and layered approach.

\subsection{Final choice}

The approach by silo seems to be a good compromise between the easiness
of calculating the m-degree of vertices and the shape of the hb-graph
by the Null vertices added and the number of elements to be filled
in the tensor: in other words $\mathcal{A}_{\text{sil}}\left(\mathcal{H}\right)$
is a good compromise between $\mathcal{A}_{\text{str}}\left(\mathcal{H}\right)$
and $\mathcal{A}_{\text{lay}}\left(\mathcal{H}\right)$. We take $\mathcal{A}_{\text{sil}}\left(\mathcal{H}\right)$
as definition of the $e$-adjacency tensor of the hb-graph. The preservation
of the information on the shape of the hb-edges through the null vertices
added allow to keep the diversity of the m-cardinality of the hb-edges.

\begin{table}
\begin{center}%
\begin{tabular}{|>{\centering}p{3.5cm}|>{\centering}m{3cm}|>{\centering}m{3cm}|>{\centering}m{3cm}|}
\cline{2-4} 
\multicolumn{1}{>{\centering}p{3.5cm}|}{} & $\mathcal{A}_{\text{str}}\left(\mathcal{H}\right)$ & $\mathcal{A}_{\text{sil}}\left(\mathcal{H}\right)$ & $\mathcal{A}_{\text{lay}}\left(\mathcal{H}\right)$\tabularnewline
\hline 
Order & $r_{\mathcal{H}}$ & $r_{\mathcal{H}}$ & $r_{\mathcal{H}}$\tabularnewline
\hline 
Dimension & $n+1$ & $n+r_{\mathcal{H}}-1$ & $n+r_{\mathcal{H}}-1$\tabularnewline
\hline 
Total number of elements & $\left(n+1\right)^{r_{\mathcal{H}}}$ & $\left(n+r_{\mathcal{H}}-1\right)^{r_{\mathcal{H}}}$ & $\left(n+r_{\mathcal{H}}-1\right)^{r_{\mathcal{H}}}$\tabularnewline
\hline 
Total number of elements potentially used by the way the tensor is
build & $\left(n+1\right)^{r_{\mathcal{H}}}$ & $\left(n+r_{\mathcal{H}}-1\right)^{r_{\mathcal{H}}}$ & $\left(n+r_{\mathcal{H}}-1\right)^{r_{\mathcal{H}}}$\tabularnewline
\hline 
Number of repeated elements per hb-edge $e_{j}=$$\left\{ v_{i_{1}}^{m_{ji_{1}}},\ldots,v_{i_{j}}^{m_{ji_{j}}}\right\} $ & $\dfrac{r_{\mathcal{H}}!}{m_{ji_{1}}!\ldots m_{ji_{j}}!n_{j}!}$

with

$n_{j}=r_{\mathcal{H}}-\#_{m}e_{j}$ & $\dfrac{r_{\mathcal{H}}!}{m_{ji_{1}}!\ldots m_{ji_{j}}!n_{jk}!}$
with

$n_{jk}=r_{\mathcal{H}}-\#_{m}e_{j}$ & $\dfrac{r_{\mathcal{H}}!}{m_{ji_{1}}!\ldots m_{ji_{j}}!}$\tabularnewline
\hline 
Number of elements to be filled per hyperedge of size $s$ before
permutation & Constant

1 & Constant

1 & Constant

1\tabularnewline
\hline 
Number of elements to be described to derived the tensor by permutation
of indices & $\left|E\right|$ & $\left|E\right|$ & $\left|E\right|$\tabularnewline
\hline 
Value of elements of a hyperedge & Dependent of hb-edge composition\\
$\dfrac{m_{ji_{1}}!\ldots m_{ji_{j}}!n_{j}!}{\left(r_{\mathcal{H}}-1\right)!}$ & Dependent of hb-edge composition\\
$\dfrac{m_{ji_{1}}!\ldots m_{ji_{j}}!n_{jk}!}{\left(r_{\mathcal{H}}-1\right)!}$ & Dependent of hb-edge composition\\
$\dfrac{m_{ji_{1}}!\ldots m_{ji_{j}}!}{\left(r_{\mathcal{H}}-1\right)!}$\tabularnewline
\hline 
Symmetric & Yes & Yes & Yes\tabularnewline
\hline 
Reconstructivity & Straightforward: delete special vertices & Straightforward: delete special vertices & Straightforward: delete special vertices\tabularnewline
\hline 
Nodes degree & Yes, but not straightforward & Yes & Yes\tabularnewline
\hline 
Spectral analysis & Special vertex increases the amplitude of the bounds  & Special vertices increase the amplitude of the bounds  & Special vertices increase the amplitude of the bounds \tabularnewline
\hline 
Interpretability of the tensor in term of hb-graph & Yes & Yes & Yes\tabularnewline
\hline 
\end{tabular}

\caption{Evaluation of the $e$-adjacency tensor depending on construction}

$\mathcal{A}_{\text{str}}\left(\mathcal{H}\right)$ refers to the
$e$-adjacency tensor built by the straightforward approach;

$\mathcal{A}_{\text{sil}}\left(\mathcal{H}\right)$ refers to the
$e$-adjacency tensor built by the silo approach;

$\mathcal{A}_{\text{lay}}\left(\mathcal{H}\right)$ refers to the
$e$-adjacency tensor built by the layered approach.

\end{center}
\end{table}

\subsection{Hypergraphs and hb-graphs}

Hypergraphs are particular case of hb-graphs and hence the $e$-adjacency
tensor defined for $e$-adjacency tensor can be used for hypergraphs.
As the multiplicity function for vertices of a hyperedge seen as hb-edge
has its values in $\left\{ 0,1\right\} $, the elements of the $e$-adjacency
tensor that differs only by a factorial due to the cardinality of
the hyperedge.

The definition that is retained for the hypergraph $e$-adjacency
tensor is:

\begin{defin}

The $e$-adjacency tensor of a hypergraph $\mathcal{H}=\left(V,E\right)$
having maximal cardinality of its hyperedges $k_{\max}$ is the tensor
$\mathcal{A}\left(\mathcal{H}\right)=\left(a_{i_{1}\ldots i_{r_{\mathcal{H}}}}\right)_{1\leqslant i_{1},\ldots,i_{r_{\mathcal{H}}}\leqslant n}$
defined by:

\[
\mathcal{A}\left(\mathcal{H}\right)=\sum\limits _{i\in\left\llbracket p\right\rrbracket }c_{e_{i}}\mathcal{R}_{e_{i}}
\]

and where for $e_{i}=\left\{ v_{j_{1}},\ldots,v_{j_{k_{i}}}\right\} \in E$
the associated tensor is: $\mathcal{R}_{e_{i}}=\left(r_{i_{1}\ldots i_{r_{\mathcal{H}}}}\right)$
, which only non-zero elements are: 
\[
r_{j_{1}\ldots j_{k_{i}}\left(n+k_{i}\right)^{k_{\max}-k_{i}}}=\dfrac{\left(k_{\max}-k_{i}\right)!}{k_{\max}!}
\]
 and all elements of $\mathcal{R}_{e_{i}}$ obtained by permuting
\[
j_{1}\ldots j_{k_{i}}\left(n+k_{i}\right)^{k_{\max}-k_{i}},
\]

and where: 
\[
c_{e_{i}}=\dfrac{k_{\max}}{k_{i}}.
\]

\end{defin}

As in \citet{ouvrard2017cooccurrence} we compare the $e$-adjacency
tensor obtained by \citet{banerjee2017spectra} and the chosen $e$-adjacency
tensor. The results are presented in Figure 

\begin{table}
\begin{center}%
\begin{tabular}{|>{\centering}p{5cm}|>{\centering}m{3.7cm}|>{\centering}m{3cm}|}
\cline{2-3} 
\multicolumn{1}{>{\centering}p{5cm}|}{} & $\mathcal{B}_{\mathcal{H}}$ & $\mathcal{\mathcal{A}\left(\mathcal{H}\right)}$\tabularnewline
\hline 
Order & $k_{\max}$ & $k_{\max}$\tabularnewline
\hline 
Dimension & $n$ & $n+k_{\max}-1$\tabularnewline
\hline 
Total number of elements & $n^{k_{\max}}$ & $\left(n+k_{\max}-1\right)^{k_{\max}}$\tabularnewline
\hline 
Total number of elements potentially used by the way the tensor is
build & $n^{k_{\max}}$ & $\left(n+k_{\max}-1\right)^{k_{\max}}$\tabularnewline
\hline 
Number of non-nul elements for a given hypergraph & $\sum\limits _{s=1}^{k_{\max}}\alpha_{s}\left|E_{s}\right|$ with\newline$\alpha_{s}=p_{s}\left(k_{\max}\right)\dfrac{k_{\max}!}{k_{1}!...k_{s}!}$  & $\sum\limits _{s=1}^{k_{\max}}\alpha_{s}\left|E_{s}\right|$ with\newline$\alpha_{s}=\dfrac{k_{\max}!}{k_{1}!...k_{s}!n_{s}!}$
with $n_{s}=k_{\max}-s$\tabularnewline
\hline 
Number of repeated elements per hyperedge of size $s$ & $\dfrac{k_{\max}!}{k_{1}!...k_{s}!}$ & $\dfrac{k_{\max}!}{k_{1}!...k_{s}!n_{s}!}$ with $n_{s}=k_{\max}-s$\tabularnewline
\hline 
Number of elements to be filled per hyperedge of size $s$ before
permutation & Varying

$p_{s}\left(k_{\max}\right)$ & Constant

1\tabularnewline
\hline 
Number of elements to be described to derived the tensor by permutation
of indices & $\sum\limits _{s=1}^{k_{\max}}p_{s}\left(k_{\max}\right)\left|E_{s}\right|$ & $\left|E\right|$\tabularnewline
\hline 
Value of elements of a hyperedge & Dependent of hyperedge composition\\
$\dfrac{s}{\alpha_{s}}$ & Dependent of hyperedge size\\
$\dfrac{\left(k_{\max}-s\right)!}{s\left(k_{\max}-1\right)!}$\tabularnewline
\hline 
Symmetric & Yes & Yes\tabularnewline
\hline 
Reconstructivity & Need computation of duplicated vertices & Straightforward: delete special vertices\tabularnewline
\hline 
Nodes degree & Yes & Yes\tabularnewline
\hline 
Spectral analysis & Yes & Special vertices increase the amplitude of the bounds \tabularnewline
\hline 
Interpretability of the tensor in term of hypergraph / hb-graph & No / No & No / Yes\tabularnewline
\hline 
\end{tabular}

\caption{Evaluation of the hypergraph $e$-adjacency tensor}

$\mathcal{B}_{\mathcal{H}}$ designates the adjacency tensor defined
in \citet{banerjee2017spectra}

$\mathcal{A\left(H\right)}$ designates the $e$-adjacency tensor
as defined in this article.

\end{center}
\end{table}

\section{Conclusion}

\label{sec:Future-work-and}

In this article, extending the concept of hypergraphs to support multisets
to hb-graphs has allowed us to define a systematic approach to built
the $e$-adjacency tensor of a hb-graph. This systematic approach
has allowed us to apply it to hypergraphs.

The tensor constructed in \citet{banerjee2017spectra} appears as
a transformation of the hypergraph $\mathcal{H}=\left(V,E\right)$
into a weighted hb-graph $\mathcal{H_{B}}=\left(V,E',w_{e}\right)$:
the hb-graph has same vertex set but the hb-edges are obtained from
the hyperedges of the original hypergraph by transforming them in
a way that for a given hyperedge all the hb-edges having this hyperedge
as support are considered with multiplicities of vertices such that
it reaches $k_{\text{max}}$.

We intend to use our new tensor in building a spectral analysis of
hypergraphs.

\section{Acknowledgements}

This work is part of the PhD of Xavier OUVRARD, done at UniGe, supervised
by Stéphane MARCHAND-MAILLET and founded by a doctoral position at
CERN, in Collaboration Spotting team, supervised by Jean-Marie LE
GOFF.

The authors are really thankful to all the team of Collaboration Spotting
from CERN.

\bibliographystyle{plainnat}
\bibliography{/home/xo/cernbox/these/000-thesis_corpus/biblio/references}

\begin{thebibliography}{26}
\providecommand{\natexlab}[1]{#1}
\providecommand{\url}[1]{\texttt{#1}}
\expandafter\ifx\csname urlstyle\endcsname\relax
  \providecommand{\doi}[1]{doi: #1}\else
  \providecommand{\doi}{doi: \begingroup \urlstyle{rm}\Url}\fi

\bibitem[Albert(1991)]{albert1991algebraic}
Joseph Albert.
\newblock Algebraic properties of bag data types.
\newblock In \emph{VLDB}, volume~91, pages 211--219. Citeseer, 1991.

\bibitem[Banerjee et~al.(2017)Banerjee, Char, and Mondal]{banerjee2017spectra}
Anirban Banerjee, Arnab Char, and Bibhash Mondal.
\newblock Spectra of general hypergraphs.
\newblock \emph{Linear Algebra and its Applications}, 518:\penalty0 14--30,
  2017.

\bibitem[Berge and Minieka(1973)]{berge1973graphs}
Claude Berge and Edward Minieka.
\newblock \emph{Graphs and hypergraphs}, volume~7.
\newblock North-Holland publishing company Amsterdam, 1973.

\bibitem[Bretto(2013)]{bretto2013hypergraph}
Alain Bretto.
\newblock Hypergraph theory.
\newblock \emph{An introduction. Mathematical Engineering. Cham: Springer},
  2013.

\bibitem[Chauve et~al.(2013)Chauve, Patterson, and
  Rajaraman]{chauve2013hypergraph}
Cedric Chauve, Murray Patterson, and Ashok Rajaraman.
\newblock Hypergraph covering problems motivated by genome assembly questions.
\newblock In \emph{International Workshop on Combinatorial Algorithms}, pages
  428--432. Springer, 2013.

\bibitem[Cooper and Dutle(2012)]{cooper2012spectra}
Joshua Cooper and Aaron Dutle.
\newblock Spectra of uniform hypergraphs.
\newblock \emph{Linear Algebra and its Applications}, 436\penalty0
  (9):\penalty0 3268--3292, 2012.

\bibitem[Grossman and Ion(1995)]{grossman1995portion}
Jerrold~W Grossman and Patrick~DF Ion.
\newblock On a portion of the well-known collaboration graph.
\newblock \emph{Congressus Numerantium}, pages 129--132, 1995.

\bibitem[Grumbach et~al.(1996)Grumbach, Libkin, Milo, and
  Wong]{grumbach1996query}
St{\'e}phane Grumbach, Leonid Libkin, Tova Milo, and Limsoon Wong.
\newblock Query languages for bags: expressive power and complexity.
\newblock \emph{ACM SIGACT News}, 27\penalty0 (2):\penalty0 30--44, 1996.

\bibitem[Hernich and Kolaitis(2017)]{hernich2017foundations}
Andr{\'e} Hernich and Phokion~G Kolaitis.
\newblock Foundations of information integration under bag semantics.
\newblock In \emph{Logic in Computer Science (LICS), 2017 32nd Annual ACM/IEEE
  Symposium on}, pages 1--12. IEEE, 2017.

\bibitem[Karypis et~al.(1999)Karypis, Aggarwal, Kumar, and
  Shekhar]{karypis1999multilevel}
George Karypis, Rajat Aggarwal, Vipin Kumar, and Shashi Shekhar.
\newblock Multilevel hypergraph partitioning: applications in vlsi domain.
\newblock \emph{IEEE Transactions on Very Large Scale Integration (VLSI)
  Systems}, 7\penalty0 (1):\penalty0 69--79, 1999.

\bibitem[Klug(1982)]{klug1982equivalence}
Anthony Klug.
\newblock Equivalence of relational algebra and relational calculus query
  languages having aggregate functions.
\newblock \emph{Journal of the ACM (JACM)}, 29\penalty0 (3):\penalty0 699--717,
  1982.

\bibitem[Lamperti et~al.(2000)Lamperti, Melchiori, and
  Zanella]{lamperti2000multisets}
Gianfranco Lamperti, Michele Melchiori, and Marina Zanella.
\newblock On multisets in database systems.
\newblock In \emph{Workshop on Membrane Computing}, pages 147--215. Springer,
  2000.

\bibitem[Loeb(1992)]{loeb1992sets}
Daniel Loeb.
\newblock Sets with a negative number of elements.
\newblock \emph{Advances in Mathematics}, 91\penalty0 (1):\penalty0 64--74,
  1992.

\bibitem[Newman(2001{\natexlab{a}})]{newman2001scientific}
Mark~EJ Newman.
\newblock Scientific collaboration networks. ii. shortest paths, weighted
  networks, and centrality.
\newblock \emph{Physical review E}, 64\penalty0 (1):\penalty0 016132,
  2001{\natexlab{a}}.

\bibitem[Newman(2001{\natexlab{b}})]{newman2001scientific-2}
Mark~EJ Newman.
\newblock Scientific collaboration networks. i. network construction and
  fundamental results.
\newblock \emph{Physical review E}, 64\penalty0 (1):\penalty0 016131,
  2001{\natexlab{b}}.

\bibitem[Ouvrard et~al.(2017)Ouvrard, {Le Goff}, and
  Marchand-Maillet]{ouvrard2017cooccurrence}
Xavier Ouvrard, Jean-Marie {Le Goff}, and Stephane Marchand-Maillet.
\newblock Adjacency and tensor representation in general hypergraphs part 1:
  e-adjacency tensor uniformisation using homogeneous polynomials.
\newblock \emph{arXiv preprint arXiv:1712.08189}, 2017.

\bibitem[Ouvrard et~al.(2018{\natexlab{a}})Ouvrard, Goff, and
  Marchand-Maillet]{ouvrard2018imaadjacency}
Xavier Ouvrard, Jean-Marie~Le Goff, and Stephane Marchand-Maillet.
\newblock On adjacency and e-adjacency in general hypergraphs: Towards a new
  e-adjacency tensor.
\newblock \emph{arXiv preprint arXiv:1809.00162}, 2018{\natexlab{a}}.

\bibitem[Ouvrard et~al.(2018{\natexlab{b}})Ouvrard, {Le Goff}, and
  Marchand-Maillet]{ouvrard2018hbgraphdiffusion}
Xavier Ouvrard, Jean-Marie {Le Goff}, and Stephane Marchand-Maillet.
\newblock Diffusion by exchanges in hb-graphs: Highlighting complex
  relationships.
\newblock \emph{CBMI Proceedings}, 2018{\natexlab{b}}.

\bibitem[P{\u a}un(2006)]{puaun2006introduction}
Gheorghe P{\u a}un.
\newblock Introduction to membrane computing.
\newblock In \emph{Applications of Membrane Computing}, pages 1--42. Springer,
  2006.

\bibitem[Radoaca(2015{\natexlab{a}})]{radoaca2015properties}
Aurelian Radoaca.
\newblock Properties of multisets compared to sets.
\newblock In \emph{Symbolic and Numeric Algorithms for Scientific Computing
  (SYNASC), 2015 17th International Symposium on}, pages 187--188. IEEE,
  2015{\natexlab{a}}.

\bibitem[Radoaca(2015{\natexlab{b}})]{radoaca2015simple}
Aurelian Radoaca.
\newblock Simple venn diagrams for multisets.
\newblock In \emph{Symbolic and Numeric Algorithms for Scientific Computing
  (SYNASC), 2015 17th International Symposium on}, pages 181--184. IEEE,
  2015{\natexlab{b}}.

\bibitem[Singh et~al.(2007)Singh, Ibrahim, Yohanna, and
  Singh]{singh2007overview}
D~Singh, A~Ibrahim, T~Yohanna, and J~Singh.
\newblock An overview of the applications of multisets.
\newblock \emph{Novi Sad Journal of Mathematics}, 37\penalty0 (3):\penalty0
  73--92, 2007.

\bibitem[Stell(2012)]{stell2012relations}
John Stell.
\newblock Relations on hypergraphs.
\newblock \emph{Relational and Algebraic Methods in Computer Science}, pages
  326--341, 2012.

\bibitem[Syropoulos(2000)]{syropoulos2000mathematics}
Apostolos Syropoulos.
\newblock Mathematics of multisets.
\newblock In \emph{Workshop on Membrane Computing}, pages 347--358. Springer,
  2000.

\bibitem[Taramasco et~al.(2010)Taramasco, Cointet, and
  Roth]{taramasco2010academic}
Carla Taramasco, Jean-Philippe Cointet, and Camille Roth.
\newblock Academic team formation as evolving hypergraphs.
\newblock \emph{Scientometrics}, 85\penalty0 (3):\penalty0 721--740, 2010.

\bibitem[Temkin et~al.(1996)Temkin, Zeigarnik, and Bonchev]{temkin1996chemical}
Oleg~N Temkin, Andrew~V Zeigarnik, and DG~Bonchev.
\newblock \emph{Chemical reaction networks: a graph-theoretical approach}.
\newblock CRC Press, 1996.

\end{thebibliography}

\end{document}